\documentclass{llncs}
\usepackage{citesort}
\usepackage{multicol}
\usepackage{amssymb}
\usepackage{amsmath,bbm}
\usepackage{latexsym,epic,eepic}
\usepackage{complexity}
\usepackage{multirow}
\usepackage{algorithmic}
\usepackage{url}
\usepackage{btran}
\usepackage{times}
\usepackage{tikz}
\usetikzlibrary{shapes}
\usetikzlibrary{automata,positioning}

\newcommand{\OMIT}[1]{}

\newcommand{\defn}[1]{\textit{#1}} 
\newcommand{\mcl}[1]{\ensuremath{\mathcal{#1}}}
\newcommand{\sem}[1]{[\![#1 ]\!]}

\newcommand{\True}{\ensuremath{\top}}

\newcommand{\struct}{\mathfrak{S}}
\newcommand{\Tree}{\ensuremath{\text{\scshape Tree}}}

\newcommand{\transys}{\ensuremath{\langle S; \to\rangle}}
\newcommand{\transysG}{\ensuremath{\langle S; \{\to\}_{a \in \ACT}\rangle}}

\newcommand{\Rec}{\ensuremath{\text{\scshape Rec}}}

\newcommand{\empseq}{\ensuremath{\epsilon}}

\newcommand{\pda}{\ensuremath{\mcl{P}}}

\newcommand{\ialphabet}{\ensuremath{\Sigma}}
\newcommand{\salphabet}{\ensuremath{\Gamma}}
\newcommand{\palphabet}{\ensuremath{\Gamma}}
\newcommand{\dalphabet}{\ensuremath{\Omega}}

\newcommand{\N}{\ensuremath{\mathbb{N}}}

\newcommand{\Z}{\ensuremath{\mathbb{Z}}}

\newcommand{\vecV}{\ensuremath{\text{\rm \bfseries v}}}

\newcommand{\problemx}[3]{
\par\noindent\underline{\sc#1}\par\nobreak\vskip.2\baselineskip
\begingroup\clubpenalty10000\widowpenalty10000
\setbox0\hbox{\bf Instance: }\setbox1\hbox{\bf Question: }
\dimen0=\wd0\ifnum\wd1>\dimen0\dimen0=\wd1\fi
\vskip-\parskip\noindent
\hbox to\dimen0{\box0\hfil}\hangindent\dimen0\hangafter1\ignorespaces#2\par
\vskip-\parskip\noindent
\hbox to\dimen0{\box1\hfil}\hangindent\dimen0\hangafter1\ignorespaces#3\par
\endgroup}

\newcommand{\ACT}{\ensuremath{\text{\sf ACT}}}

\newcommand{\Aut}{\ensuremath{\mathcal{A}}} 
\newcommand{\transrel}{\ensuremath{\Delta}} 
\newcommand{\tran}[1]{\ensuremath{\stackrel{#1}{\longrightarrow}}}
\newcommand{\controls}{\ensuremath{Q}} 
\newcommand{\finals}{\ensuremath{F}} 
\newcommand{\Lang}{\ensuremath{\mathcal{L}}} 
\newcommand{\Parikh}[1]{\ensuremath{\mathbb{P}(#1)}} 
\newcommand{\Label}[1]{\ensuremath{\text{\textsc{lab}}(#1)}} 
\newcommand{\AutRun}{\ensuremath{\rho}} 

\newcommand{\arity}{\ensuremath{\text{\texttt{ar}}}}

\newcommand{\rarw}{\rightarrow}

\newcommand{\ModelRun}{\ensuremath{\pi}}

\newcommand{\Size}[1]{\ensuremath{\|#1\|}}

\newcommand{\FINAL}{\ensuremath{\text{\scshape FINAL}}}
\newcommand{\Distinct}{\ensuremath{\text{\scshape distinct}}}
\newcommand{\Edge}{\ensuremath{\text{\scshape edge}}}

\newcommand{\DEC}{\ensuremath{\text{\texttt{DEC}}}}
\newcommand{\INC}{\ensuremath{\text{\texttt{INC}}}}
\newcommand{\x}{\ensuremath{\text{\texttt{x}}}}
\newcommand{\y}{\ensuremath{\text{\texttt{y}}}}

\newcommand{\permgroup}{\ensuremath{\mathcal{S}}}
\newcommand{\rotgroup}{\ensuremath{\mathcal{R}}}
\newcommand{\dihgroup}{\ensuremath{\mathcal{D}}}

\newcommand{\padding}{\ensuremath{\#}}

\newcommand{\biword}[2]{(#1,#2)}

\newif\ifdraft\draftfalse
\ifdraft
\newcommand\anthony[1]{{\color{blue}
[#1 - \textbf{Anthony}]}}

\newcommand\philipp[1]{{\color{magenta}
[#1 - \textbf{Philipp}]}}

\newcommand\sunjun[1]{{\color{brown}
[#1 - \textbf{Sun Jun}]}}

\newcommand\khanh[1]{{\color{purple}
[#1 - \textbf{Khanh}]}}

\newcommand\anthonychanged[1]{{\color{blue}{#1}}}

\else
\newcommand\todo[1]{}
\newcommand\anthony[1]{}
\newcommand\philipp[1]{}
\newcommand\khanh[1]{}
\newcommand\sunjun[1]{}
\newcommand\anthonychanged[1]{#1}

\fi

\newcommand\shortlong[2]{#2}

\usepackage{color}
\shortlong{}{\pagestyle{plain}}

\title{Regular Symmetry Patterns (Technical Report)}
\author{
    Anthony W. Lin\inst{1} \and
    Truong Khanh Nguyen\inst{2} \and
    Philipp R\"ummer\inst{3} \and
    Jun Sun\inst{4}
}
\institute{
    Yale-NUS College, Singapore \and
    Autodesk, Singapore \and
    Uppsala University, Sweden \and
    Singapore University of Design and Technology
}
\date{}

\begin{document}

\maketitle

\begin{abstract}
    Symmetry reduction is a well-known approach for alleviating the state explosion 
problem in model checking. Automatically identifying 
symmetries in concurrent systems, however, is computationally
expensive. 
We propose a symbolic framework for capturing symmetry patterns in
parameterised systems (i.e. an infinite family of finite-state systems):
two regular word transducers to represent, respectively, parameterised systems 
and symmetry patterns. The framework subsumes various 
types of ``symmetry
relations'' ranging from weaker notions (e.g. simulation preorders) to the 
strongest notion (i.e. isomorphisms). Our framework enjoys two algorithmic
properties: (1) symmetry verification: given a transducer, 
we can automatically check whether it is a symmetry pattern of a given 
system, and (2) symmetry synthesis: we can automatically generate 
a symmetry pattern for a given system in the form of a transducer.
Furthermore, our symbolic language allows additional constraints that the 
symmetry patterns need to satisfy to be easily incorporated in the
verification/synthesis. We show how these properties can help identify 
symmetry patterns in examples like dining philosopher protocols, 
self-stabilising protocols, and prioritised resource-allocator protocol.
In some cases (e.g. Gries's coffee can problem), our technique automatically
synthesises a safety-preserving finite approximant, which can then be
verified for safety solely using a finite-state model checker.

\end{abstract}

\section{Introduction}
\label{sec:intro}
Symmetry reduction~\cite{CJEF96,Ip1996,ES96} is a well-known approach for 
alleviating the state explosion problem in automatic verification of concurrent
systems. 
The essence of symmetry reduction is to identify symmetries in the system and 
avoid exploring states that are ``similar'' (under these symmetries) to 
previously explored states. 

One main challenge with symmetry reduction methods is the 
difficulty in identifying symmetries in a given system in general. 
One approach is to provide dedicated language instructions for specifying 
symmetries (e.g. see \cite{Ip1996,Sistla2000,Spermann2008}) or specific 
languages (e.g. see \cite{Jaghoori2005,Jaghoori2010,Murphi}) so that users can 
provide 
insight on 
what symmetries are there in the system. For instance, Mur$\varphi$ 
provides a 
special data type with a list of syntactic restrictions and all values that 
belong to this type are symmetric. Another approach is to detect 
symmetry automatically without requiring expert insights. Automatic detection 
of symmetries 
is an extremely difficult computational problem. A number of approaches have 
been proposed in this direction (e.g. \cite{Donaldson2005,Donaldson2008,ASE13}).
For example, Donaldson and Miller \cite{Donaldson2005,Donaldson2008} designed an
automatic approach to detecting process symmetries for channel-based 
communication systems, based on constructing a graph called \emph{static channel
diagram} from a Promela model whose automorphisms correspond to symmetries in 
the model. 
Nonetheless, it is
clear from their experiments that existing approaches work only for small 
numbers of processes.

In practice, concurrent systems are often obtained by replicating a generic 
behavioral description \cite{WD10}. For example, a prioritised 
resource-allocator protocol \cite[Section 4.4]{donaldson-thesis}
provides a description of an allocator program and a client program in
a network with a star topology (allocator in the center), from which
a concurrent system with $1$ allocator and $m$ clients (for any given $m
\in \Z_{> 0}$) can be generated. This is in fact the standard setting of
parameterised systems (e.g. see \cite{vojnar-habilitation,rmc-survey}),
which are symbolic descriptions of infinite families
$\{\struct_i\}_{i=1}^{\infty}$ of transition systems $\struct_i$ that can be
generated by instantiating some parameters
(e.g. the number of processes).


Adopting this setting of parameterised systems, we consider the
problem of formulating and generating symbolic \emph{symmetry
  patterns}, abstract descriptions of symmetries that can be
instantiated to obtain concrete symmetries for every instance of a
parameterised system.  A formal language to specify symmetry patterns
should be able to capture interesting symmetry patterns, e.g., that
each instance $\struct_i$ of the parameterised system $\struct =
\{\struct_i\}_{i=1}^{\infty}$ exhibits the full symmetry $S_n$
(i.e. invariant under permuting the locations of the processes).
Ideally, such a language $\mathcal{L}$ should also enjoy the following 
algorithmic properties: (1)
\emph{symmetry verification}, i.e., given a symmetry pattern $P \in 
\mathcal{L}$, 
we can automatically check whether $P$ is a symmetry pattern of a given 
parameterised system, and (2) \emph{symmetry synthesis}: given a
parameterised system, we can automatically generate
symmetry patterns~$P \in \mathcal{L}$ that the system exhibits.
In particular, if $\mathcal{L}$ is sufficiently expressive to specify commonly 
occuring symmetry patterns, Property (1) would allow us to 
automatically compute which common symmetry patterns hold for a given 
parameterised system. 
In the case when symmetry patterns might be less obvious, Property (2) would
allow us to identify further symmetries that are satisfied by the given 
parameterised systems.
To the best of our knowledge, to date no such languages have been proposed.

%


\OMIT{
For this reason, the problem of automatically generalising 
``symmetry patterns'' of \emph{every} instance of a given parameterised system from 
symmetries for \emph{small} instances of the system is highly relevant for
symmetry reduction. Despite this, we are not aware of existing general techniques
for checking (let alone, discovering) symmetry patterns for parameterised
systems, even in the simple case of rotation symmetry, i.e., supposing that we 
found that instances of a parameterised system with $2,\ldots,5$ processes 
exhibit the rotation symmetry, does \emph{every} instance of the system
have the rotation symmetry?
}

\vspace{-.6cm}
\subsubsection*{Contribution:}
We propose a general symbolic framework for capturing 
symmetry patterns for parameterised systems.
The framework uses \emph{finite-state letter-to-letter word transducers} to 
represent \emph{both} parameterised systems and symmetry patterns. 
In the sequel, symmetry patterns
that are recognised by transducers are called \emph{regular symmetry
patterns}.
Based on extensive studies in regular model checking 
(e.g. see \cite{Abdulla12,rmc-survey,vojnar-habilitation,Lin12b}),
finite-state word transducers are now well-known to be good symbolic 
representations of parameterised systems. 
Moreover, equivalent logic-based (instead of automata-based) formalisms
are also available, e.g., LTL(MSO) 
\cite{LTLMSO} which can be used to specify parameterised systems and
properties (e.g. safety and liveness) in a convenient way.
In this paper, we show that
transducers are not only also sufficiently expressive for representing 
many common symmetry patterns, but they enjoy the two aforementioned desirable 
algorithmic properties: automatic symmetry verification and synthesis. 

There is a broad spectrum of notions of ``symmetries'' for transition systems
that are of interest to model checking. These include 
simulation preorders (a weak variant) and isomorphisms (the strongest), e.g., 
see \cite{BK08}. We suggest that transducers are not only sufficiently 
powerful in expressing many such notions of symmetries, but they are also a 
flexible 
symbolic language in that constraints (e.g. the symmetry pattern is a bijection)
can be easily added to or relaxed from the specification.
In this paper, we shall illustrate this point by handling
simulation preorders and isomorphisms (i.e. bijective simulation preorders)
within the same framework.
Another notable point of our symbolic language is its ability to specify
that the simulation preorder gives rise to an \emph{abstracted system} that is 
finite-state and preserves non-safety (i.e. if the original system is not safe, 
then so is the abstracted system). In other words, \emph{we can specify that the
symmetry pattern reduces the infinite-state parameterised system to a
finite-state system}. Safety of finite-state systems can then be checked using 
standard finite-state model checkers.

%
%

We next show how to specialise our framework to
\emph{process symmetries} \cite{CJEF96,Ip1996,ES96}. 
Roughly speaking, a process
symmetry for a concurrent system $\struct$ with $n$ processes is a permutation 
$\pi: [n] \to [n]$ (where $[n] := \{1,\ldots,n\}$) such that the behavior of 
$\struct$
is invariant under permuting the process indices by $\pi$ (i.e. the resulting
system is isomorphic to the original one under the natural bijection 
induced by $\pi$). For example, if the process indices of clients in the 
aforementioned resource-allocator protocol with $1$ allocator and $m$ clients
are $1, \ldots, m+1$, then any permutation $\pi: [m+1] \to [m+1]$ that fixes 
$1$ is a process symmetry for the protocol. The set of such process symmetries 
is a permutation group on $[m+1]$ (under functional composition) generated by 
the following two permutations specified in standard cyclic notations: $(2,3)$ 
and $(2,3,\ldots,m+1)$. This is true
for \emph{every value of $m \geq 2$}. In addition, finite-state model checkers
represent symmetry permutation groups by their (often exponentially
more succinct) finite set of generators.
Thus, if $\struct = \{\struct_n\}_{n=1}^{\infty}$ is a parameterised system
where $\struct_n$ is the instance with $n$ processes, we represent the
\emph{parameterised symmetry groups} $\mathcal{G} = 
\{G_n\}_{n=1}^{\infty}$ 
(where $G_n$ is the process symmetry group for $\struct_n$) by a finite list of 
regular symmetry patterns that generate $\mathcal{G}$.
We postulate that commonly occuring parameterised process symmetry groups (e.g. 
full symmetry groups and rotations groups) can be captured in this framework, 
e.g.,
parameterised symmetry groups for the aforementioned resource-allocator 
protocol can be generated by the symmetry patterns
$\{(2,3)(4)\cdots (m+1)\}_{m \geq 3}$ and
$\{(2,3,\ldots,m+1)\}_{m \geq 3}$,
which can be easily expressed using transducers. 
Thus, using our symmetry verification
algorithm, commonly occuring process symmetries for a given parameterised 
system could be automatically identified.

The aforementioned approach of checking a given parameterised system
against a ``library'' of common regular symmetry patterns has two problems.
Firstly, some common symmetry patterns are not regular, e.g., reflections. To 
address this, we equip our transducers with an unbounded pushdown stack.
Since pushdown transducers in general cannot be synchronised \cite{ABL07}
(a crucial property to obtain our symmetry verification algorithm), we
propose a restriction of pushdown transducers for which we can recover
automatic symmetry verification. Secondly, there are many useful but subtle
symmetry patterns in practice. To address this, we propose the use of our 
symmetry synthesis algorithm. Since a naive enumeration of all transducers with
$k=1,\ldots,n$ states does not scale, we devise a CEGAR loop for our algorithm
in which 
a SAT-solver provides a candidate symmetry pattern (perhaps satisfying some
extra constraints) and an automata-based algorithm either verifies the
correctness of the guess, or returns a counterexample that can be further
incorporated into the guess of the SAT-solver.

We have implemented our symmetry verification/synthesis algorithms and 
demonstrated its usefulness in identifying regular symmetry patterns for 
examples like dining philosopher protocols, self-stabilising protocols, 
resource-allocator protocol, and Gries's coffee can problem. 
In the case of the coffee can problem, we managed to obtain a reduction from 
the infinite system to a finite-state system.

\vspace{-.5cm}
\subsubsection*{Related Work:}
Our work is inspired by regular model checking (e.g. 
\cite{rmc-survey,vojnar-habilitation,Abdulla12,LTLMSO}), which focuses on 
symbolically 
computing the sets of reachable configurations of parameterised systems as 
regular languages.
Such methods are generic, but are not
guaranteed to terminate in general. As in regular model checking, our framework
uses transducers to represent parameterised systems. However, instead of 
computing their sets of reachable configurations, our work finds symmetry 
patterns of the parameterised systems, which can be
exploited by an explicit-state finite-state model checker to verify the desired 
property over finite instances of the system (see \cite{WD10} for more details).
Although our verification algorithm is guaranteed to terminate in general
(in fact, in polynomial-time assuming the parameterised system is given as a
DFA), our synthesis algorithm only terminates when
we fix the number of states for the transducers. 
Finding process symmetry patterns is often easier since there are 
available
tools for finding symmetries for finite (albeit small) instances of the systems
(e.g. \cite{Donaldson2008,Donaldson2005,ASE13}).

Another related line of works is ``cutoff techniques'' (e.g. see
\cite{EN95,EK00} and the survey \cite{vojnar-habilitation}), which allows one to
reduce verification of parameterised systems into verification of finitely
many instances (in some cases, $\leq 10$ processes). These works usually assume 
verification of ${\sf LTL}\backslash{\sf X}$ properties. Although such 
techniques are extremely powerful, the systems that can be handled
using the techniques are often quite specific (e.g. see 
\cite{vojnar-habilitation}).

\OMIT{
problem is 
Mention the orbit problem and some solutions.
}

\vspace{-.5cm}
\subsubsection*{Organisation:}
Section~\ref{sec:prelim} contains preliminaries.
In Section~\ref{sec:framework}, we present our framework of regular symmetry
patterns.
In Section~\ref{sec:verification} (resp. Section~\ref{sec:onesym}), we present 
our symmetry verification algorithm (resp. synthesis) algorithms.
Section~\ref{sec:expr} discusses our implementation and experiment results. 
Section~\ref{sec:conc} concludes with future work. 
Due to space constraints, some details are relegated into the 
\shortlong{full version}{appendix}.

\section{Preliminaries}
\label{sec:prelim}

\noindent
\textbf{General notations.} 
For two given natural numbers $i \leq j$, we define $[i,j] = \{i,i+1,
\ldots,j\}$. Define $[k] = [1,k]$.
Given a set $S$, we use $S^*$ to denote the set of all finite
sequences of elements from $S$. The set $S^*$ always includes the empty
sequence which we denote by $\empseq$. 
Given two sets of words
$S_1, S_2$, we use $S_1\cdot S_2$ to denote the set $\{ v\cdot w\mid v\in S_1,
w\in S_2\}$ of words formed by concatenating words from $S_1$ with words from
$S_2$. Given two relations $R_1,R_2 \subseteq S \times S$, we define
their composition as $R_1 \circ R_2 = \{ (s_1,s_3) \mid \exists s_2.\; (s_1,s_2)
\in R_1 \wedge (s_2,s_3) \in R_2\}$. Given a subset $X \subseteq S$, we
define the image $R(X)$ (resp. preimage $R^{-1}(X)$) of $X$ under $R$ as 
the set $\{ s \in S \mid \exists s'.\; (s',s) \in R \}$ (resp.
$\{ s' \in S \mid \exists s.\; (s',s) \in R \}$).
Given a finite set~$S =
\{s_1, \ldots, s_n\}$, the \emph{Parikh vector}~$\Parikh{v}$ of a word
$v \in S^*$ is the vector~$(|v|_{s_1}, \ldots, |v|_{s_n})$ of the number
of occurrences of the elements~$s_1, \ldots, s_n$, respectively, in $v$.
\smallskip

\noindent
\textbf{Transition systems} Let $\ACT$ be a finite set of
\defn{action symbols}. A \defn{transition system} over $\ACT$ is
a tuple $\struct = \transysG$,
where $S$ is a set of \defn{configurations}, and $\to_a\ \subseteq S \times S$ 
is a binary relation over $S$. 
We use $\to$ to denote the relation $\left(\bigcup_{a \in \ACT} \to_a\right)$. 
In the sequel, we will often only consider the case when $|\ACT| = 1$ for 
simplicity.
The notation $\to^+$ (resp. $\to^*$) is used to denote the transitive (resp.
transitive-reflexive) closure of $\to$. 
We say that a sequence $s_1 \to \cdots \to s_n$ is a \defn{path} (or
\defn{run}) in $\struct$ (or in $\to$). Given two paths $\ModelRun_1: s_1
\to^* s_2$ and $\ModelRun_2: s_2 \to^* s_3$ in $\to$, we may concatenate them
to obtain $\ModelRun_1 \odot \ModelRun_2$ (by gluing together $s_2$). In
the sequel, for each $S' \subseteq S$ we use the notation $post^*_{\to}(S')$ to 
denote the set of configurations $s \in S$ reachable in $\struct$ from some 
$s \in S$.
\OMIT{
Given a relation $\to \subseteq S \times S$ and subsets 
$S_1,\ldots,S_n \subseteq S$, denote by 
$\Rec_{\to}(\{S_i\}_{i=1}^n)$ to be the set of elements $s_0 \in S$ for which
there exists an infinite path $s_0 \to s_1 \to \cdots$ visiting
each $S_i$ infinitely often, i.e., such that, for each
$i\in[1,n]$, there are infinitely many $j \in \N$ with $s_j \in S_i$.
}
\smallskip

\noindent
\textbf{Words, automata, and transducers.}
We assume basic familiarity with word automata.
Fix a finite alphabet $\Sigma$. For each finite word $w = w_1\ldots w_n \in 
\Sigma^*$, we
write $w[i,j]$, where $1 \leq i \leq j \leq n$, to denote the segment
$w_i\ldots w_j$. Given a (nondeterministic finite) automaton $\mcl{A} = 
(\Sigma,Q,\delta,q_0,F)$,
a run of $\mcl{A}$ on $w$ is a function $\rho: \{0,\ldots,n\}
\rarw Q$ with $\rho(0) = q_0$ that obeys the transition relation $\delta$. 
We may also denote the run $\rho$ by the word $\rho(0)\cdots \rho(n)$ over
the alphabet $Q$. 
The run $\rho$ is said to be \defn{accepting} if $\rho(n) \in F$, in which
case we say that the word $w$ is \defn{accepted} by $\mcl{A}$. The language
$L(\mcl{A})$ of $\mcl{A}$ is the set of words in $\Sigma^*$ accepted by
$\mcl{A}$. In the sequel, we will use the standard abbreviations DFA/NFA
(Deterministic/Nondeterministic Finite Automaton).

Transducers are automata that accept 
binary relations over words \cite{Blum99,BG04} (a.k.a. ``letter-to-letter'' 
automata, or synchronised transducers). 
Given two words
$w = w_1\ldots w_n$ and $w' = w_1'\ldots w_m'$ over the alphabet $\ialphabet$, 
let $k = \max\{n,m\}$ and $\ialphabet_\padding := \ialphabet \cup \{\padding\}$,
where $\padding$ is a special padding symbol not in $\ialphabet$.
We define a word $w \otimes w'$ of length $k$ over alphabet
$\ialphabet_{\padding} 
\times \ialphabet_{\padding}$ as follows:
$$ w \otimes w' \ = \ \biword{a_1}{b_1}\ldots\biword{a_k}{b_k}, \
  \text{where}\
a_i \ = \ \begin{cases} w_i & i \leq n \\
                         \padding &  i > n,
            \end{cases}
\ \text{and} \ 
b_i \ = \ \begin{cases} w_i' & i \leq m \\
                         \padding & i > m.
            \end{cases}
$$
In other words, the shorter word is padded with $\padding$'s, and the
$i$th letter of $w \otimes w'$ is then the pair of the $i$th letters of
padded $w$ and $w'$. 
A \defn{transducer} (a.k.a. letter-to-letter automaton) is simply a finite-state
automaton over 
$\ialphabet_{\padding}\times\ialphabet_{\padding}$, and a binary relation $R
\subseteq \ialphabet^* \times \ialphabet^*$ is \defn{regular} if
the set $\{ w \otimes w' : (w,w') \in R \}$ is accepted by
a letter-to-letter automaton. 
The relation $R$ is said to be \defn{length-preserving} if $R$ only relates 
words of the same length \cite{rmc-survey}, i.e., that any automaton recognising
$R$ consumes no padded letters
of the form $(a,\padding)$ or $(\padding,a)$.
In the sequel, for notation simplicity,
we will confuse a transducer and the binary relation
that it recognises (i.e. $R$ is used to mean both).

Finally, notice that the notion of regular relations can be easily extended
to $r$-ary relations $R$ for each positive integer $r$ (e.g. see
\cite{Blum99,BG04}). 
To this end, the input alphabet of the transducer will be $\ialphabet_{\padding}^r$.
Similarly, for $R$ to be regular, the set $\{ w_1 \otimes \cdots \otimes w_r :
(w_1,\ldots,w_r) \in R\}$ of words over the alphabet $\ialphabet^r$ must be
regular. 

\smallskip
\noindent
\textbf{Permutation groups.}
We assume familiarity with basic group theory (e.g. see \cite{group-book}).
A \defn{permutation} on $[n]$ is any bijection $\pi: [n] \to [n]$. The set of
all permutations on $[n]$ forms the \defn{($n$th) full symmetry group}
$\permgroup_n$ under functional composition. A \defn{permutation group} on
$[n]$ is any set of permutations on $[n]$ that is a subgroup of $\permgroup_n$
(i.e. closed under composition). A \defn{generating set} for a permutation
group $G$ on $[n]$ is a finite set $X$ of permutations (called
\emph{generators}) such that each permutation in $G$ can be expressed by taking 
compositions of elements in $X$. In this case, we say that $G$ can be generated
by $X$. 
A word $w = a_0\ldots a_{k-1} \in [n]^*$ containing 
distinct elements of $[n]$ (i.e. $a_i \neq a_j$ if $i \neq j$) can be used
to denote the permutation that maps $a_i \mapsto a_{i+1 \mod k}$ for
each $i \in [0,k)$ and fixes other elements of $[n]$. In this case, $w$ is
    called a \defn{cycle} (more precisely, $k$-cycle or \emph{transposition} in
    the case when $k=2$), which we will often write in the 
standard notation $(a_0,\ldots,a_{k-1})$ so as to avoid confusion. 
Any permutation can be written as a composition of disjoint cycles
\cite{group-book}. In addition, it is known that $\permgroup_n$ can be 
generated by the set $\{(1,2),(1,2,\ldots,n)\}$. Each subgroup $G$ 
of $\permgroup_n$ acts on the 
set $\Sigma^n$ (over any finite alphabet $\Sigma$) under the group action of
permuting indices, i.e., for each $\pi \in G$ and $\vecV = (a_1,\ldots,a_n) \in 
\Sigma^n$, we define $\pi\vecV := (a_{\pi^{-1}(1)},\ldots,a_{\pi^{-1}(n)})$. 
That way, each $\pi$ induces the bijection $f_\pi: \Sigma^n \to \Sigma^n$
such that $f_\pi(\vecV) = \pi\vecV$.

Given a permutation group $G$ on $[n]$ and a transition system
$\struct = \transys$ with state space $S = \Sigma^n$, we say that $\struct$ is
\emph{$G$-invariant} if the bijection $f_\pi: \Sigma^n \to \Sigma^n$ induced 
by each $\pi \in G_n$ is an automorphism on $\struct$, i.e., 
$\forall v,w \in S$: $v \to w$ implies $f_\pi(v) \to f_\pi(w)$.

\section{The formal framework} \label{sec:framework}

This section describes our symbolic framework regular symmetry patterns. 
\subsection{Representing parameterised systems}
As is standard in regular model checking 
\cite{vojnar-habilitation,rmc-survey,Abdulla12},
we use length-preserving transducers to represent parameterised systems. 
As we shall see below, we will use non-length-preserving transducers to 
represent symmetry patterns.
\begin{definition}[Automatic transition systems\footnote{Length-preserving 
            automatic transition systems are instances of automatic structures 
        \cite{Blum99,BG04}}]
    A transition system $\struct =\linebreak \transysG$
    is said to be \defn{(length-preserving) automatic} if $S$ 
    is a regular set over a finite alphabet $\ialphabet$ and each 
    relation $\to_a$ is given by a transducer 
    over $\ialphabet$.
\end{definition}
More precisely, the parameterised 
system defined by $\struct$ is the family $\{\struct_n\}_{n \geq 0}$ with 
\mbox{$\struct_n = \langle S_n; \to_{a,n} \rangle$}, 
where $S_n := S \cap \Sigma^n$ is the set of all words in $S$ of length $n$ and 
$\to_{a,n}$ is the transition relation $\to_a$ restricted to $S_n$.
In the sequel, for simplicity we will mostly consider examples when $|\ACT| = 
1$. When the meaning is understood, we shall confuse the notation $\to_a$ for 
the transition relation of $\struct$ and the transducer that recognises it.
To illustrate our framework and methods, we shall give three examples of 
automatic transition systems 
(see \cite{LTLMSO,vojnar-habilitation} for numerous other examples).
\OMIT{
For our purposes, we partition the finite alphabet $\ialphabet$ of
our transducers
into a \defn{process alphabet} $\palphabet$ and a \defn{data alphabet} $\dalphabet$.
Here, the data alphabet is used to represent the values of global variables in
the system.
For simplicity, we will always assume that $S = \palphabet^*\dalphabet^*$, i.e., 
values of global variables are placed at the end of the word representations of
configurations. 
\begin{remark}
    There are more expressive ways of integrating data and process values, e.g., 
    by using the product alphabet 
$\palphabet \times \dalphabet$. They, however, often require the use of a larger 
alphabet size and more states, which is not desirable from an implementation 
standpoint. In 
addition, our convention is sufficient to represent virtually all natural examples 
of parameterised systems.
\end{remark}
}

\begin{example}
    \label{ex:resource}
    We describe a prioritised resource-allocator protocol
    \cite[Section 4.4]{donaldson-thesis}, which is a simple
    mutual exclusion protocol in network with a star topology. The protocol
    has one allocator and $m$ clients. Initially, each process is in an
    \emph{idle} state. However, clients might from time to time \emph{request}
    for an access to a resource (\emph{critical section}), which can only be 
    used by one process at
    a time. For simplicity, we will assume that there is only one resource
    shared by all the clients. The allocator manages the use of the resource.
    When a request is lodged by a client, the allocator can allow the
    client to use the resource. When the client has finished using
    the resource, it will send a message to the allocator, which can then
    allow other clients to use the resource.

    To model the protocol as a transducer, we let $\ialphabet = \{i,r,c\}$,
    where $i$ stands for ``idle'', $r$ for ``request'', and $c$ for
    ``critical''. Allocator can be in either the state $i$ or the state
    $c$, while a client can be in one of the three states in $\ialphabet$.
    A valid configuration is a word $aw$, where $a \in \{i,c\}$ represents
    the state of the allocator and $w \in \ialphabet^*$ represents the
    states of the $|w|$ clients (i.e. each position in $w$ represents
    a state of a client). Letting $I = \{(a,a) : a \in \ialphabet\}$
    (representing idle local transitions), the transducer can be described by a 
    union of the following regular expressions:
\begin{itemize}
    \item $I^+(i,r)I^*$ --- a client requesting for a resource.
    \item $(i,c)I^*(r,c)I^*$ ---  a client request granted by the allocator.
    \item $(c,i)I^*(c,i)I^*$ --- the client has finished using the resource. 
        \qed
\end{itemize}
\end{example}

\begin{example}
\label{ex:token}
We describe Israeli-Jalfon self-stabilising protocol \cite{IJ90}.
The original protocol is probabilistic, but since we are only interested in 
reachability, we may use nondeterminism to model randomness. The protocol has a 
ring topology, and each process either holds a token (denoted by $\top$) or
does not hold a token (denoted by $\bot$). Dynamics is given by
the following rules:
\begin{itemize}
\item A process $P$ holding a token can pass the token to either the left or the
    right neighbouring process $P'$, provided that $P'$ does not hold a token.
\item If two neighbouring processes $P_1$ and $P_2$ hold tokens, the tokens
    can be merged and kept at process $P_1$.
\end{itemize}
We now provide a transducer that formalises this parameterised
system. Our relation is on words over the alphabet $\Sigma =
\{\bot,\top\}$, and thus a transducer is an automaton that runs over
$\Sigma \times \Sigma$. In the following, we use $I :=
\{\biword{\top}{\top},\biword{\bot}{\bot}\}$.  The automaton is given by
a union of the following regular expressions:
\begin{multicols}{3}
\begin{itemize}
\item $I^* \biword{\top}{\bot}\biword{\bot}{\top} I^*$
\item $I^* \biword{\bot}{\top}\biword{\top}{\bot} I^*$
\item $I^* \biword{\top}{\top}\biword{\top}{\bot} I^*$
\item $\biword{\bot}{\top}I^*\biword{\top}{\bot}$
\item $\biword{\top}{\bot}I^*\biword{\bot}{\top}$
\item $\biword{\top}{\bot}I^*\biword{\top}{\top}$
\qed
\end{itemize}
\end{multicols}
\end{example}

\begin{example}
    \label{ex:gries}
    Our next example is the classical David Gries's coffee can problem, which uses
    two (nonnegative) integer variables $x$ and $y$ to store the number of black
    and white coffee beans, respectively. At any given step, if $x+y \geq 2$ 
    (i.e. there are at least two coffee beans), then two coffee beans are 
    nondeterministically chosen. First, if both are of the same colour,
    then they are both discarded and a new black bean is put in the can. 
    Second, if they are of a different colour, the white bean is kept and the black 
    one is discarded. We are usually interested in the colour of the last bean
    in the can. We formally model Gries's coffee can problem as a transition
    system with domain $\N \times \N$ and transitions:
    \begin{enumerate}
        \item[(a)] if $x \geq 2$, then $x := x - 1$ and $y := y$.
        \item[(b)] if $y \geq 2$, then $x := x + 1$ and $y := y-2$.
        \item[(c)] if $x \geq 1$ and $y \geq 1$, then $x := x - 1$ and $y := y$.
    \end{enumerate}
    To distinguish the colour of the last bean, we shall add self-loops to
    all configurations in $\N \times \N$, except for the configuration $(1,0)$.
    We can model the system as a length-preserving transducer as follows. The 
    alphabet is 
    $\ialphabet := \dalphabet_\x \cup \dalphabet_\y$, where $\dalphabet_\x :=
    \{ 1_\x, \bot_\x \}$ and $\dalphabet_\y := \{ 1_\y, \bot_\y\}$. 
    A configuration is a word in the regular language
    $1_\x^*\bot_\x^*1_\y^*\bot_\y^*$.
    For example, the configuration with $x = 5$ and $y = 3$, where the maximum
    size of the integer buffers $\x$ and $\y$ is 10, is represented as the word 
    $(1_\x)^5(\bot_\x)^5(1_\y)^3(\bot_\y)^7$. The transducer for the 
    coffee can problem can be easily constructed.
  \qed

  \OMIT{
    \begin{figure}
\begin{center}
$\psmatrix[colsep=1.5cm,rowsep=1.5cm,mnode=circle]
& p_1 & p_2 & 
\endpsmatrix
\psset{nodesep=1pt}
\psset{arrowscale=1.5}
\ncline{->}{1,1}{1,2}
$
\caption{Some basic automata components for the coffee can 
problem.\label{fig:coffee}}
\end{center}
\end{figure}
}
\end{example}

\subsection{Representing symmetry patterns}
\begin{definition}
  Let $\struct = \transys$ be a transition system with $S \subseteq
  \ialphabet^*$. A \defn{symmetry pattern} for
  $\struct= \transys$ is a simulation preorder $R \subseteq 
  S \times S$ for $\struct$, i.e., satisfying:
\begin{description}
    \item[(S1)] $R$ respects each $\to_a$, i.e., for all $v_1,v_2,w_1 \in S$,
        if $v_1 \to_a w_1$, and $(v_1,v_2) \in R$, then there exists
        $w_2 \in S$ such that $(w_1,w_2) \in R$ and $v_2 \to_a w_2$;
    \item[(S2)] $R$ is \emph{length-decreasing}, i.e., for all $v_1,v_2 \in S$,
        if $(v_1,v_2) \in R$, then $|v_1| \geq |v_2|$.
\end{description}
The symmetry pattern is said to be \defn{complete} if additionally the relation 
is length-preserving and a bijective function.
\end{definition}
Complete symmetry patterns will also be denoted by functional notation~$f$.
In the case of complete symmetry pattern $f$, it can be observed 
that Condition (S1) also entails that $f(v) \to_a f(w)$ implies $v \to_a w$. 
This condition does not hold in general for simulation preorders. We shall
also remark that, owing to the well-known property of simulation preorders,
symmetry patterns preserve non-safety. To make this notion more precise,
we define the image of a transition system $\struct = \transys$ (with 
$S \subseteq \ialphabet^*$) under the symmetry pattern $R$ as the transition 
system $\struct_1 = \langle S_1; \to_1 \rangle$ such that
$S_1 = R(S)$ and that $\to_1$ is the restriction of $\to$ to $S_1$.
\begin{proposition}
    Given two sets $I, F \subseteq \ialphabet^*$, if $post_{\to_1}^*(R(I)) 
    \cap R(F) = \emptyset$, then $post_{\to}^*(I) \cap
    F = \emptyset$.
\end{proposition}
In other words, if $\struct_1$ is safe, then so is $\struct$. In the case
when $\struct_1$ is finite-state, this check can be performed using a
standard finite-state model checker. We shall define now a class of symmetry
patterns under which the image $\struct_1$ of the input transition system
can be automatically computed.
\begin{definition}[Regular symmetry pattern]
    A symmetry pattern $R \subseteq S \times S$ for an automatic transition 
    system~$\struct = \transys$ is said to be \defn{regular} if the relation
    $R$ is regular.
    %
\end{definition}
\begin{proposition}
    Given an automatic transition system $\struct = \transys$ (with
    $S \subseteq \ialphabet^*$) and a regular symmetry pattern $R \subseteq
    S \times S$, the image of $\struct$ under $R$ is an automatic
    transition system and can be constructed in polynomial-time.
    \label{prop:proj}
\end{proposition}
In particular, whether the image of $\struct$ under $R$ is a finite system
can be automatically checked since checking whether the language
of an NFA is finite can be done in polynomial-time.
The proof of this proposition (in the \shortlong{full version}{appendix})
is a simple automata 
construction that relies on the fact that regular relations are closed under 
projections. We shall next illustrate the concept of regular symmetry patterns 
in action,
    especially for Israeli-Jalfon self-stabilising protocol and Gries's coffee
    can problem.

    We start with Gries's coffee can problem (cf. Example \ref{ex:gries}).
    Consider the function $f: (\N \times \N) \to (\N \times \N)$ where
    $f(x,y)$ is defined to be (i) $(0,1)$ if $y$ is odd, (ii) $(2,0)$ if $y$ is 
    even and $(x,y) \neq (1,0)$, and (iii) $(1,0)$ if $(x,y) = (1,0)$.
    \OMIT{
    \[
        f(x,y) := \left\{ \begin{array}{ll}
                        (0,1)  & \text{\quad if $y$ is odd} \\
                (2,0)  & \text{\quad if $y$ is even and $(x,y) \neq (1,0)$} \\
                        (1,0)  & \text{\quad if $(x,y) = (1,0)$}
                          \end{array}
                  \right.
    \]
    }
    This is a symmetry pattern since the last bean for the coffee can problem 
    is white iff $y$ is odd. Also, that a configuration 
    $(x,y)$ with $y \equiv 0 \pmod{2}$ and $x > 1$ is mapped to $(2,0)$ is 
    because $(2,0)$ has a self-loop, while $(1,0)$ is a dead end. 
    It is easy to show that $f$ is a regular symmetry pattern. To this end,
    we
    construct a transducer for each of the cases (i)--(iii).
    For example, the transducer handling the case $(x,y)$ when $y \equiv 1 
    \pmod{2}$ works as follows: simultaneously read the pair
    $(v,w)$ of words and ensure that $w = \bot_\x\bot_\x1_y$ and
    $v \in 1_\x^*\bot_\x^*1_\y(1_\y1_\y)^*\bot_\y^*$.
    \OMIT{
    \begin{center}
    \begin{tikzpicture}[%
    >=stealth,
    shorten >=1pt,
    node distance=2cm,
    on grid,
    auto,
    state/.append style={minimum size=2em},
    thick
  ]
    \node[state] (A)              {};
    \node[state] (D)      [right of=A] {};
    \node[state] (E)      [right of=E] {};
    \node[state,accepting] (B) [right of=D] {};
    \node[state] (C) [right of=B] {};

    \path[->] (A) +(-1,0) edge (A)
              (A)         edge              node {$\{1_\x,\bot_\x\}/\bot_\x$}(D)
              (D)         edge              node {$\{1_\x,\bot_\x\}/\bot_\x$}(E)
              (E)         edge [loop above] node {$\{1_\x,\bot_\x\}/\bot_\x$} ()
              (A)         edge              node {$1_\y/1_\y$} (B)
              (B)         edge [bend left]  node {$\{1_\y,\bot_\y\}/\bot_\y$} (C)
              (C)         edge [bend left]  node {$\{1_\y,\bot_\y\}/\bot_\y$} (B);
  \end{tikzpicture}
  \end{center}
  }
  As an NFA, the final transducer has $\sim 10$ states. 

\paragraph{Process symmetry patterns.}
We now apply the idea of regular symmetry patterns to capture process
symmetries in parameterised systems. We shall show how this applies to
Israeli-Jalfon self-stabilising protocol.
A \defn{parameterised permutation} is a family $\bar\pi = \{\pi_n\}_{n \geq 1}$ 
of permutations $\pi_n$ on $[n]$. We say that $\bar\pi$ is \defn{regular} if, 
for each alphabet $\ialphabet$, the bijection $f_{\bar\pi}: \ialphabet^* \to 
\ialphabet^*$ defined by $f_{\bar\pi}(\vecV) := \pi_n\vecV$, where $\vecV \in 
\ialphabet^n$, is a regular relation. We say that $\bar\pi$ is \defn{effectively
regular} if $\bar\pi$ is regular and if there is an algorithm which,
on input $\ialphabet$, constructs a transducer for the bijection $f_{\bar\pi}$.
As we shall only deal with effectively regular permutations, when understood
we will omit mention of the word ``effectively''.
As we shall see below, examples of effectively regular parameterised 
permutations include transpositions (e.g. $\{(1,2)(3)\cdots (n)\}_{n \geq 2}$) 
and rotations $\{(1,2,\ldots,n)\}_{n \geq 1}$. 

We now extend the notion of parameterised permutations to 
\defn{parameterised symmetry groups} $\mathcal{G} := \{G_n\}_{n \geq 1}$
for parameterised systems, i.e., each $G_n$ is a permutation group on $[n]$.
A finite set $F = \{\bar\pi^1,\ldots,\bar\pi^r\}$ of parameterised permutations 
(with $\bar\pi^j = \{\pi_n^j\}_{n \geq 1}$)
\defn{generates} the parameterised symmetry groups $\mathcal{G}$ if each group 
$G_n \in \mathcal{G}$ can be generated by the set $\{ \pi_n^j : j \in [r]\}$,
i.e., the $n$th instances of parameterised permutations in $F$. We say that
$\mathcal{G}$ is \defn{regular} if each $\bar\pi^j$ in $F$ is regular. 

We will single out three 
commonly occuring process symmetry groups for concurrent systems with $n$ 
processes: full symmetry group $\permgroup_n$ (i.e. generated by $(1,2)$ and 
$(1,2,\ldots,n)$), rotation group 
$\rotgroup_n$ (i.e. generated by $(1,2,\ldots,n)$), and the dihedral group 
$\dihgroup_n$ (i.e. generated by $(1,2,\ldots,n)$ and the ``reflection'' 
permutation $(1,n)(2,n-1)\cdots (\lfloor n/2 \rfloor,\lceil n/2\rceil)$). 
The parameterised versions of them are: (1) $\permgroup := \{\permgroup_n\}_{n
\geq 1}$, (2) $\rotgroup := \{\rotgroup_n\}_{n \geq 1}$, and (3)
$\dihgroup := \{\dihgroup_n\}_{n \geq 1}$. 
\begin{theorem}
    Parameterised full symmetry groups $\permgroup$ and parameterised rotation 
    symmetry groups $\rotgroup$ are effectively regular.
    \label{th:paramreg}
\end{theorem}
As we will see in Proposition~\ref{prop:dihedralreg} below, parameterised dihedral groups
are not regular. We will say how to deal with this in the next section.
As we will see in Theorem \ref{th:invariantreg}, Theorem \ref{th:paramreg} can 
be used to construct a fully-automatic method for checking whether 
\emph{each} instance $\struct_n$ of a parameterised system $\struct =
\{\struct_n\}_{n \geq 0}$ represented by a given transducer $\Aut$ has a
full/rotation process symmetry group. 
\begin{proof}[sketch of Theorem \ref{th:paramreg}]
    To show this, it suffices to show that 
    $\mathcal{F} = \{(1,2)(3)\cdots (n)\}_{n\geq 2}$
    and
    $\mathcal{F}' = \{(1,2,\ldots,n)\}_{n\geq 2}$ are effectively
    regular. [The degenerate case when $n = 1$ can be handled easily if
    necessary.] For, if this is the case, then the parameterised full symmetry 
    $\permgroup$ and the parameterised rotation symmetry groups can be 
    generated by (respectively) $\{\mathcal{F},\mathcal{F}'\}$ and 
    $\mathcal{F}'$. Given an input $\ialphabet$, the transducers for both 
    $\mathcal{F}$ and $\mathcal{F}'$ are easy. For example, the transducer
    for $\mathcal{F}$ simply swaps the first two letters in the input, i.e.,
    accepts pairs of words of the form $(abw,baw)$ where $a,b \in \ialphabet$ 
    and $w \in \ialphabet^*$. These transducers can be constructed in 
    polynomial time (details in the \shortlong{full version}{appendix}). \qed
\end{proof}
The above proof shows that $\{(1,2)(3)\cdots (n)\}_{n \geq 0}$
and $\{(1,2,\ldots,n)\}_{n \geq 0}$ are regular parameterised permutations.
Using the same proof techniques, we can also show that the
following simple variants of these parameterised permutations are also regular
for each $i \in \Z_{> 0}$: (a) $\{(i,i+1)(i+2)\cdots (n)\}_{n \geq 1}$, and (b) 
$\{(i,i+1,\ldots,n)\}_{n \geq 1}$.
As we saw from Introduction, the prioritised resource-allocator protocol
has a star topology and so both $\{(2,3)(4)\cdots \}_{n \geq 1}$ and
$\{(2,3,\ldots,n)\}_{n \geq 1}$ generate complete symmetry patterns for
the protocol (i.e. invariant under permuting the clients). 
Therefore, our library $\mathcal{L}$ of regular symmetry patterns 
could store all of these 
regular parameterised permutations (up to some fixed $i$). 

    Parameterised dihedral groups $\dihgroup$ are generated by 
    rotations $\bar\pi = \{(1,2,\ldots,n)\}_{n \geq 2}$ and
    reflections $\bar\sigma = 
    \{(1,n)(2,n-1)\cdots (\lfloor n/2\rfloor,\lceil n/2\rceil)\}_{n
    \geq 2}$. Reflections $\bar\sigma$ are, however, not regular for the same
    reason that the language of palindromes (i.e. words that are the same
    read backward as forward). In fact, it is not possible to find a different 
    list of generating parameterised permutations that are regular (proof in 
    the \shortlong{full version}{appendix}):

\begin{proposition}
    Parameterised dihedral groups $\dihgroup$ are not regular.
    \label{prop:dihedralreg}
\end{proposition}

\section{Symmetry verification}
\label{sec:verification}
In this section, we will present our symmetry verification algorithm for
regular symmetry patterns. We then show how to extend the algorithm to a more 
general framework of symmetry patterns that subsumes parameterised dihedral
groups.
\OMIT{
given an automatic transition system $\struct = \transys$ and a
regular relation $R \subseteq S \times S$, check if $R$ is a symmetry pattern
for $\struct$. In the case when $R$ is a function (which often suffices in
practice), we show that this can be done in polynomial time assuming that
the automaton for $\struct$ is given as a DFA ($R$ can be given as an NFA). 
In addition, each of the following checks can be done automatically and 
independently from $\struct$: given 
regular relation $R \subseteq S \times S$ (1) check whether $R$ is a function, 
and (2) $R$ is a bijection.
}

\subsection{The algorithm}

\begin{theorem}
Given an automatic transition system $\struct = \transys$ and a
regular relation $R \subseteq S \times S$, we can automatically check if $R$ is 
a symmetry pattern of $\struct$. 
\label{th:verification-exp}
\end{theorem}
\begin{proof}
    Let $D$ be the set of words over the alphabet $\ialphabet^3$ of the form
    $v_1 \otimes v_2 \otimes w_1$, for some words $v_1, v_2, w_2 \in 
    \ialphabet^*$ satisfying: (1) $v_1 \to w_1$, (2) $(v_1,v_2) \in R$, and  
    (3) there does \emph{not} exist $w_2 \in \ialphabet^*$ such that $v_2 \to 
    w_2$ and $(w_1,w_2) \in R$.
    Observe that $R$ is a symmetry pattern for $\struct$ iff $D$ is
    empty. An automaton $\Aut = (\ialphabet^3,\controls,\transrel,q_0,F)$
    for $D$ can be constructed via a classical automata construction.

    As before, for simplicity of presentation, we will assume that $S = 
    \ialphabet^*$; for, otherwise, we can perform a simple product automata
    construction with the automaton for $S$.
    Let $\Aut_1 = (\ialphabet^2,\controls_1,\transrel_1,q_0^1,F_1)$ be
    an automaton for $\to$, and 
    $\Aut_2 = (\ialphabet_\padding^2,\controls_2,\transrel_2,q_0^2,F_2)$ an 
    automaton for $R$. 
    
    We first construct an NFA
    $\Aut_3 = (\ialphabet_\padding^2,\controls_3,\transrel_3,q_0^3,F_3)$ for
    the set $Y \subseteq S \times S$ consisting of pairs $(v_2,w_1)$ such that
    the condition (3) above is \emph{false}. This can be done by a simple
    product/projection automata construction that takes into account the fact
    that $R$ might not be length-preserving: That is, define $\controls_3 := 
    \controls_1 \times \controls_2$, $q_0^3 := (q_0^1,q_0^2)$, and 
    $F_3 := F_1 \times F_2$.
    The transition relation $\transrel$ consists of transitions
            $((q_1,q_2),(a,b),(q_1',q_2'))$ such that, for some
            $c \in \ialphabet_\padding$, it is the case that 
            $(q_2,(b,c),q_2') \in \transrel_2$ and one of the following
            is true: (i) $(q_1,(a,c),q_1') \in \transrel_1$, (ii) 
            $q_1 = q_1'$, $b \neq \padding$, and $a = c = \padding$.
    Observe that the construction for $\Aut_3$ runs in polynomial-time.

    In order to construct $\Aut$, we will have to perform a complementation
    operation on $\Aut_3$ (to compute the complement of $Y$) and apply 
    a similar product automata construction. The former takes exponential
    time (since $\Aut_3$ is nondeterministic), while the latter costs 
    an additional polynomial-time overhead. \qed
\end{proof}

The above algorithm runs in exponential-time even if $R$ and $\struct$
are presented as DFA, since an automata projection operation in general yields
an NFA. The situation improves dramatically when $R$ is 
\emph{functional} (i.e. for all $x \in S$, there exists a unique $y \in S$ such 
that $R(x,y)$).\philipp{same should hold for partial functions}
\begin{theorem}
There exists a polynomial-time algorithm which, given an automatic transition
system $\struct = \transys$ presented as a DFA and a functional regular relation
$R \subseteq S \times S$ presented as an NFA, decides whether 
$R$ is a symmetry pattern for $\struct$.
\label{th:verification-poly}
\end{theorem}
\begin{proof}
    Let $D$ be the set of words over the alphabet $\ialphabet^4$ of the form 
    $v_1 \otimes v_2 \otimes w_1 \otimes w_2$, for some words
    $v_1,v_2,w_1,w_2 \in \ialphabet^*$ satisfying: (1) $v_1 \to w_1$,
        (2) $(v_1,v_2) \in R$, (2') $(w_1,w_2) \in R$, and (3) $v_2 \not\to w_2$
    Observe that $R$ is a symmetry pattern for $\struct$ iff $D$ is
    empty. The reasoning is similar to the proof of Theorem
    \ref{th:verification-exp}, but the difference now is that given any
    $w_1 \in \ialphabet^*$, there is a \emph{unique} $w_2$ such that
    $(w_1,w_2) \in R$ since $R$ is functional. For this reason, we need only
    to make sure that $v_2 \not\to w_2$. An automaton $\Aut$ 
    for $D$ can be constructed by first complementing the automaton for $\to$ 
    and then a standard product automata construction as before. The latter
    takes polynomial-time if $\to$ is presented as a DFA, while the latter
    costs an additional polynomial-time computation overhead (even if
    $R$ is presented as an NFA). \qed
\end{proof}

\begin{proposition}
    The following two problems are solvable in polynomial-space: 
    given a regular relation $R \subseteq S \times S$, check whether
    (1) $R$ is functional, and (2) $R$ is a bijective function. Furthermore,
    the problems are polynomial-time reducible to language inclusion for NFA.
    \label{prop:bijective}
\end{proposition}
Observe that there are fast heuristics for checking language inclusion for NFA 
using antichain and simulation techniques (e.g. see 
\cite{abdulla-antichain,BP13}). The proof of the above proposition uses
standard automata construction, which is relegated to the 
\shortlong{full version}{appendix}.

\subsection{Process symmetries for concurrent systems}
We say that an automatic transition system $\struct = \transys$ (with $S
\subseteq \ialphabet^*$) is $\mathcal{G}$-\defn{invariant} if each instance
$\struct_n = \langle S \cap \Gamma^n; \to \rangle$ of $\struct$ is
$G_n$-invariant. If $\mathcal{G}$ is generated by regular parameterised
permutations $\bar \pi^1,\ldots,\bar \pi^r$, then $\mathcal{G}$-invariance is 
equivalent to the condition that, for each $j \in [r]$, the bijection 
$f_{\pi^j}: \ialphabet^* \to \ialphabet^*$ is a regular symmetry pattern 
for $\struct$. The following theorem is an immediate corollary of Theorem
\ref{th:verification-poly}.
\begin{theorem}
    Given an automatic transition system $\struct = \transys$ (with $S \subseteq
    \ialphabet^*$) and a regular parameterised symmetry group $\mathcal{G}$ 
    presented
    by regular parameterised permutations $\bar\pi^1,\ldots,\bar\pi^k$,
    we can check that $\struct$ is $\mathcal{G}$-invariant in polynomial-time
    assuming that $\struct$ is presented as DFA.
    \label{th:invariantreg}
\end{theorem}
In fact, to check whether $\struct$ is $\mathcal{G}$-invariant, it suffices
to sequentially go through \emph{each} $\bar\pi^j$ and ensure that it is 
a symmetry pattern for $\struct$, which by Theorem \ref{th:verification-poly}
can be done in polynomial-time.

\subsection{Beyond regular symmetry patterns}
Proposition \ref{prop:dihedralreg} tells us that regular symmetry patterns
do not suffice to capture parameterised reflection permutation. This leads us
to our inability to check whether a parameterised system is invariant 
under parameterised dihedral symmetry groups, e.g., Israeli-Jalfon's
self-stabilising protocol and other randomised protocols including
Lehmann-Rabin's protocol (e.g. \cite{LR81}).
To deal with this problem, we extend the
notion of regular length-preserving symmetry patterns to a subclass of 
``context-free'' symmetry patterns that preserves some nice algorithmic 
properties. \emph{Proviso: All relations considered in this subsection are 
length-preserving.}

Recall that a \defn{pushdown automaton (PDA)} is a tuple 
$\pda = (\ialphabet,\salphabet,\controls,\transrel,q_0,\finals)$, where
$\ialphabet$ is the input alphabet, $\salphabet$ is the stack alphabet
(containing a special bottom-stack symbol, denoted by $\bot$, that cannot 
be popped),
$\controls$ is the finite set of control states, 
$q_0 \in \controls$ is an initial state, $\finals \subseteq \controls$ is a 
set of final states, and $\transrel
\subseteq (\controls \times \salphabet) \times \ialphabet \times 
(\controls \times \salphabet^{\leq 2})$ is a set of transitions, where
$\salphabet^{\leq 2}$ denotes the set of all words of length at most 2.
A configuration of $\pda$ is a pair $(q,w) \in \controls \times \salphabet^*$
with \defn{stack-height} $|w|$. For each $a \in \ialphabet$, we define 
the binary relation $\to_a$ on configurations of $\pda$ as follows:
$(q_1,w_1) \to_a (q_2,w_2)$ if there exists a transition
$((q_1,o),a,(q_2,v)) \in \transrel$ such that $w_1 = wo$ and $w_2 = wv$
for some $w \in \salphabet^*$. A \defn{computation path} $\pi$ of $\pda$ on 
input $a_1\ldots a_n$ is any sequence
\[
    (q_0,\bot) \to_{a_1} (q_1,w_1) \to_{a_2} \cdots \to_{a_n} 
    (q_n,w_n)
\]
of configurations from the initial state $q_0$. In the following, the 
\defn{stack-height sequence} of $\pi$ is the sequence 
    $|\bot|, |w_1|, \ldots, |w_n|$ of stack-heights.
We say that a computation path $\pi$ is \defn{accepting} if $q_n \in \finals$.

We now extend Theorem \ref{th:invariantreg} to a class of transducers that 
allows us to capture the reflection symmetry. This class consists of
``height-unambiguous'' pushdown transducers, which is a subclass of pushdown 
transducers that is amenable to synchronisation. We say that a pushdown
automaton is \defn{height-unambiguous (h.u.)} if it satisfies the restriction 
that the stack-height sequence in an \emph{accepting} computation path on 
an input word $w$ is uniquely determined by the length $|w|$ of $w$.
That is, given an accepting computation path $\pi$ on $w$ and an accepting
computation path $\pi'$ of $w'$ with $|w| = |w'|$, the stack-height 
sequences of $\pi$ and $\pi'$ coincide. Observe that the definition allows
the stack-height sequence of a non-accepting path to differ. A language
$L \subseteq \ialphabet^*$ is said to be \defn{height-unambiguous context-free
(huCF)} if it is recognised by a height-unambiguous PDA. A simple 
example of a huCF language 
is the language of palindromes (i.e. the input word is the same backward
as forward). A simple non-example of a huCF language 
is the language of well-formed nested parentheses.
This can be proved by a standard pumping argument.

We extend the definitions of regularity of length-preserving relations, 
symmetry patterns, etc.
from Section \ref{sec:prelim} and Section \ref{sec:framework} to 
height-unambiguous pushdown automata in the obvious way, e.g., a
length-preserving relation $R \subseteq S \times S$ is \defn{huCF} if 
$\{ v \otimes w : (v,w) \in R \}$ is a huCF language. We saw in Proposition
\ref{prop:dihedralreg} that parameterised dihedral symmetry groups $\dihgroup$
are not regular. We shall show now that they are huCF.
\begin{theorem}
    Parameterised dihedral symmetry groups $\dihgroup$ are effectively 
    height-unambiguous context-free.
    \label{th:dihedralhuCF}
\end{theorem}
\begin{proof}
    To show this, it suffices to show that the parameterised reflection
    permutation $\bar\sigma = \{\sigma_n\}_{n \geq 2}$, where 
    $\sigma_n := (1,n)(2,n-1)\cdots (\lfloor n/2\rfloor,\lceil n/2\rceil)$,
    is huCF. To this end,
    given an input alphabet $\ialphabet$, we construct a PDA
    $\pda = (\ialphabet^2,\salphabet,\controls,\transrel,q_0,\finals)$
    that recognises $f_{\bar\sigma}: \ialphabet^* \to \ialphabet^*$
    such that $f_{\bar\sigma}(\vecV) = \sigma_n\vecV$ whenever
    $\vecV \in \ialphabet^n$. The PDA $\pda$ works just like the PDA
    recognising the language of palindromes. We shall first give the intuition.
    Given a word $w$ of the form $v_1 \otimes v_2 \in (\ialphabet^2)^*$, we 
    write $w^{-1}$ to denote the word $v_2 \otimes v_1$. On an input word
    $w_1w_2w_3 \in (\ialphabet^2)^*$, where $|w_1| = |w_3|$ and
    $|w_2| \in \{0,1\}$, the PDA will save $w_1$ in the stack and compares
    it with $w_3$ ensuring that $w_3$ is the reverse of $w_1^{-1}$. It will also
    make sure that $w_2 = (a,a)$ for some $a \in \ialphabet$ in the case
    when $|w_2|  = 1$. The formal definition of $\pda$ is given in the 
    \shortlong{full version}{appendix}.
    \qed
\end{proof}


\begin{theorem}
There exists a polynomial-time algorithm which, given an automatic transition
system $\struct = \transys$ presented as a DFA and a functional h.u.
context-free relation
$R \subseteq S \times S$ presented as an NFA, decides whether $R$ is a
symmetry pattern for $\struct$.
\label{th:verification-huCF}
\end{theorem}
To prove this theorem, let us revisit the automata construction from the proof 
of Theorem \ref{th:verification-poly}. The problematic part of the construction
is that we need to show that,
given an huCF relation $R$, the 4-ary relation 
\begin{equation}
    \mathcal{R} := (R \times R) \cap \{ (w_1,w_2,w_3,w_4) \in (\ialphabet^*)^4: |w_1| = |w_2| 
    = |w_3| = |w_4| \}
    \tag{$*$}
\end{equation}
is also huCF. The 
rest of the construction requires only taking product with regular relations
(i.e. $\to$ or its complement), which works for unrestricted pushdown automata
since context-free languages are closed under taking product with regular 
languages via the usual product automata construction for regular languages.

\begin{lemma}
    Given an huCF relation $R$, we can construct in polynomial-time an 
    h.u. PDA recognising the 4-ary relation $\mathcal{R}$.
    \label{lm:4ary}
\end{lemma}
\begin{proof}
    Given a h.u. PDA 
    $\pda = (\ialphabet^2,\salphabet,\controls,\transrel,q_0,\finals)$
    recognising $R$, we will construct a PDA
    $\pda' = (\ialphabet^4,\salphabet',\controls',\transrel',q_0',\finals')$
    recognising $\mathcal{R}$. Intuitively, given an input $(v,w) \in
    \mathcal{R}$, the PDA $\pda'$ is required to run two copies of $\pda$ at 
    the same time, one on the input $v$ (to check that $v \in R$) and the other 
    on input $w$ (to check that $w \in R$). Since $\pda$ is height-unambiguous
    and $|v| = |w|$, we can assume that the stack-height sequences of 
    accepting runs of $\pda$ on $v$ and $w$ coincide. That is, in an accepting
    run $\pi_1$ of $\pda$ on $v$ and an accepting run of $\pi_2$ of $\pda$
    on $w$, when a symbol is pushed onto (resp. popped from) the 
    stack at a certain position in $\pi_1$, then a symbol is also pushed onto 
    (resp. popped from) the stack in the same position in $\pi_2$. The converse
    is also true. These two stacks can, therefore, be simultaneously simulated 
    using only one stack of $\pda'$ with $\salphabet' = \salphabet \times
    \salphabet$. For this reason, the rest of the details is a standard product 
    automata construction for finite-state automata. Therefore, the
    automaton $\pda'$ is of size quadratic in the size of $\pda$.
    The detailed definition of $\pda'$ is given in the 
    \shortlong{full version}{appendix}.
    \qed
\end{proof}

\OMIT{
\begin{theorem}
    Given an automatic transition system $\struct = \transys$ (with $S \subseteq
    \ialphabet^*$) and a parameterised symmetry group $\mathcal{G}$ 
    presented by height-unambiguous context-free parameterised permutations 
    $\bar\pi^1,\ldots,\bar\pi^k$,
    we can check that $\struct$ is $\mathcal{G}$-invariant in polynomial-time
    assuming that $\struct$ is presented as DFA.
    \label{th:invariantcfl}
\end{theorem}
}

\OMIT{
The following proposition shows that the class of huCF languages is syntactic,
i.e., there exists an algorithm for enumerating them. To show this, since
the class of context-free languages is syntactic (since we can simply
enumerate PDAs), it suffices to give an algorithm for checking whether a given
PDA is height-unambiguous.
\begin{proposition}
    Given a pushdown automaton $\mcl{P}$, it is decidable to check
    whether $\mcl{P}$ is h.u. \anthony{Say the complexity}
\end{proposition}
\begin{proof}

\end{proof}
}

We shall finally pinpoint a limitation of huCF symmetry patterns,
and discuss how we can address the problem in practice. 
It can be proved by a simple reduction from Post Correspondence
Problem that it is undecidable to check whether a given PDA is
height-unambiguous. In practice, however, this is not a major obstacle since it 
is possible
to \emph{manually} (or \emph{semi-automatically}) add a selection of huCF 
symmetry patterns to our library $\mathcal{L}$ of regular symmetry patterns
from Section \ref{sec:framework}. Observe that this effort is \emph{independent}
of any parameterised system that one needs to check for symmetry. 
Checking whether any huCF symmetry pattern in $\mathcal{C}$ is a symmetry 
pattern for a given automatic transition system $\struct$ can then be done 
automatically and efficiently (cf. Theorem \ref{th:verification-huCF}).
For example, Theorem \ref{th:dihedralhuCF} and Theorem 
\ref{th:verification-huCF} imply that we can automatically check whether an 
automatic transition system is invariant under the parameterised dihedral 
groups:
\begin{theorem}
    Given an automatic transition system $\struct = \transys$ (with $S \subseteq
    \ialphabet^*$) presented as DFA, checking whether $\struct$ is 
    $\dihgroup$-invariant can be done in polynomial-time.
\end{theorem}
Among others,
this allows us to automatically confirm that Israeli-Jalfon self-stabilising
protocol is $\dihgroup$-invariant.

\section{Automatic Synthesis of Regular Symmetry Patterns}
\label{sec:onesym}

Some regular symmetry patterns for a given automatic system might not be
obvious, e.g., Gries's coffee can example. Even in the case of process
symmetries, the user might choose different representations for the same
protocol. For example, the allocator process in Example~\ref{ex:resource} 
could be represented by the last (instead of the first) letter in the word, 
which would mean that $\{(1,2,\ldots,n-1)\}_{n \geq 3}$ and
$\{(1,2)(3)\cdots (n)\}_{n \geq 3}$ are symmetry patterns for the 
system (instead of $\{(2,3,\ldots,n)\}_{n \geq 2}$ and
$\{(2,3)(4)\cdots (n)\}_{n \geq 3}$). Although we can put reasonable variations
of common symmetry patterns in our library~$\mathcal{L}$, we would benefit
from a systematic way of synthesising regular symmetry patterns for a given
automatic transition system~$\struct$. In this section, we will describe our 
automatic technique for achieving this. We focus on the case of symmetry
patterns that are \emph{total functions} (i.e.\ homomorphisms), but the
approach can be generalised to other patterns.
\OMIT{The technique is a CEGAR-loop that
involves a SAT-solver for providing a candidate regular symmetry pattern
and an automata-based method for verifying the correctness of the guess,
or returns a counterexample, which will be further incorporated into the
SAT-solver to make a better guess.}

\OMIT{
We 
focus on the task of synthesising a \emph{single} symmetry for a
system; the approach can be generalised to the generation of multiple
elements of the symmetry group, but this is beyond the scope of this
paper. \anthony{The last sentence sounds annoying}
}

Every transducer~$\mcl{A} = (\ialphabet_\padding \times
\ialphabet_\padding,Q,\delta,q_0,F)$ over $\ialphabet_\padding^*$ represents a
regular binary relation~$R$ over $\ialphabet^*$. We have shown in
Section~\ref{sec:verification} that we can automatically check whether $R$ represents
a symmetry pattern, perhaps satisfying further constraints like functionality
or bijectivity as desired by the user. Furthermore, we can also automatically 
check that it is a symmetry pattern for a given automatic transition
system~$\struct$.
Our overall approach for computing such transducers makes use of two main 
components, which are performed iteratively within a refinement loop:
\begin{description}
\item[\textsc{Synthesise}] A candidate transducer~$\Aut$ with $n$
  states is computed with the help of a SAT-solver, enforcing a
  relaxed set of conditions encoded as a Boolean constraint~$\psi$
  (Section~\ref{sec:guessing}). 
\item[\textsc{Verify}] As described in Section~\ref{sec:verification},
  it is checked whether the binary relation~$R$ represented by $\Aut$
  is a symmetry pattern for $\struct$ (satisfying further constraints
  like completeness, as desired by the user).
  If this check is negative, $\psi$ is strengthened to eliminate
  counterexamples, and \textsc{Synthesise} is invoked
  (Section~\ref{sec:counterexamples}).
\end{description}

This refinement loop is enclosed by an outer loop that increments the
parameter~$n$ (initially set to some small number $n_0$) when
\textsc{Synthesise} determines that no transducers satisfying $\psi$
exist anymore. The next sections describe the \textsc{Synthesise}
step, and the generation of counterexamples in case \textsc{Verify}
fails, in more detail.

\subsection{\textsc{Synthesise}: Computation of a Candidate Transducer $\Aut$}
\label{sec:guessing}

Our general encoding of transducers $\mcl{A} = (\ialphabet_\padding \times
\ialphabet_\padding,Q,\delta,q_0,F)$ uses a representation as a deterministic
automaton (DFA), which is suitable for our refinement loop since
counterexamples (in particular, words that should not be accepted) can
be eliminated using succinct additional constraints.  We assume that
the states of the transducer $\mcl{A}$ to be computed are $\controls =
\{1,\ldots,n\}$, and that $q_0 = 1$ is the initial state.  We use the
following variables to encode transducers with $n$ states:
\begin{itemize}
\item $x_{t}$ (of type Boolean), for each tuple $t = (q,a,b,q') \in \controls
  \times \ialphabet_\padding \times \ialphabet_\padding \times \controls$;
\item $z_q$ (of type Boolean), for each $q \in \controls$.
\end{itemize}
The assignment~$x_t = 1$ is interpreted as the existence of the
transition~$t$ in $\Aut$. Likewise, we use $z_q = 1$ to represent that
$q$ is an accepting state in the automaton; since we use DFA, it is in
general necessary to have more than one accepting state.

The set of considered transducers in step~\textsc{Synthesise} is
restricted by imposing a number of conditions, selected depending on
the kind of symmetry to be synthesised: for general symmetry
homomorphisms, conditions~\textbf{(C1)}--\textbf{(C8)} are used, for
\emph{complete} symmetry patterns \textbf{(C1)}--\textbf{(C10)}, and
for \emph{process symmetries} \textbf{(C1)}--\textbf{(C11)}.
\begin{description}
\item[(C1)] The transducer~$\Aut$ is deterministic.
\item[(C2)] For every transition~$q \tran{\biword{a}{b}} q'$ in $\Aut$
  it is the case that $a \not= \padding$.\footnote{Note that all
    occurrences of ~$\padding$ are in the end of words.}
\item[(C3)] Every state of the transducer is reachable from the
  initial state.
\item[(C4)] From every state of the transducer an accepting state can
  be reached.
\item[(C5)] The initial state~$q_0$ is accepting.
\item[(C6)] The language accepted by the transducer is infinite.
\item[(C7)] There are no two transitions $q \tran{\biword{a}{b}} q'$
  and $q \tran{\biword{a}{b'}} q'$ with $b \neq b'$.
\item[(C8)] If an accepting state~$q$ has self-transitions~$q
  \tran{\biword{a}{a}} q$ for every letter~$a \in \ialphabet_\padding$, then $q$
  has no outgoing edges.
  \bigskip
\item[(C9)] For every transition~$q \tran{\biword{a}{b}} q'$ in $\Aut$
  it is the case that $b \not= \padding$.
\item[(C10)] There are no two transitions $q \tran{\biword{a}{b}} q'$
  and $q \tran{\biword{a'}{b}} q'$ with $a \neq a'$.
\end{description}
Condition~\textbf{(C2)} implies that computed transducers are
length-decreasing, while \textbf{(C3)} and \textbf{(C4)} rule out
transducers with redundant states.  \textbf{(C5)} and \textbf{(C6)}
follow from the simplifying assumption that only homomorphic
symmetries patterns are computed, since a transducer representing a
total function~$\ialphabet^* \to \ialphabet^*$ has to accept the empty
word and words of unbounded length. Note that \textbf{(C5)} and
\textbf{(C6)} are necessary, but not sufficient conditions for total
functions, so further checks are needed in \textsc{Verify}.
\textbf{(C7)} and \textbf{(C8)} are necessary (but again not sufficient)
conditions for transducers representing total functions, given the
additional properties \textbf{(C3)} and \textbf{(C4)}; it can be shown
that a transducer violating \textbf{(C7)} or \textbf{(C8)} cannot be a
total function.
Condition~\textbf{(C9)} implies that padding~$\padding$ does not occur
in any accepted word, and is a sufficient condition for
length-preservation; as a result, the symbol~$\padding$ can be
eliminated altogether from the transducer construction.

Finally, for \emph{process symmetries} the assumption can be
made that the transducer preserves not only word length, but also the
number of occurrences of each symbol:
\begin{description}
\item[(C11)] The relation~$R$ represented by the transducer only
  relates words with the same Parikh vector, i.e., $R(v, w)$
  implies $\Parikh{v} = \Parikh{w}$.
\end{description}

The encoding of the conditions~\textbf{(C1)}--\textbf{(C11)}
as Boolean constraints is mostly straightforward.
Further
Boolean constraints can be useful in special cases, in particular for
Example~\ref{ex:gries} the restriction can be made that only
\emph{image-finite} transducers are computed. 
We can also
constrain the search in the \textsc{Synthesise} stage to those
transducers that accept manually defined words $W = \{v_1 \otimes w_1,
\ldots, v_k \otimes w_k\}$, using a similar encoding as the one for
counterexamples in Sect.~\ref{sec:counterexamples}. This technique can be
used, among others, to systematically search for
symmetry patterns that generalise some known finite symmetry.

\subsection{Counterexample Generation}
\label{sec:counterexamples}

Once a transducer $\Aut$ representing a candidate relation $R
\subseteq \Sigma^* \times \Sigma^*$ has been computed,
Theorem~\ref{th:verification-exp} can be used to implement the
\textsc{Verify} step of the algorithm. When using the construction
from the proof of Theorem~\ref{th:verification-exp}, one of three
possible kinds of counterexample can be detected, corresponding to
three different formulae to be added to the constraint~$\psi$ used in
the \textsc{Synthesise} stage:
\begin{enumerate}
\item \makebox[0.67\linewidth][l]{A word~$v$ has to be included in the
    domain~$R^{-1}(\Sigma^*_\padding)$:} $\exists w.~R(v, w)$
\item \makebox[0.67\linewidth][l]{A word~$w$ has to be included in the
    range~$R(\Sigma^*_\padding)$:} $\exists v.~R(v, w)$
\item \makebox[0.67\linewidth][l]{One of two contradictory pairs has
    to be eliminated:} $\neg R(v_1, w_1) \vee \neg R(v_2, w_2)$
\end{enumerate}
Case~1 indicates relations~$R$ that are not total; case~2 relations
that are not surjective; and case~3 relations that are not functions,
not injective, or not simulations.\footnote{Note that this is for the
  special case of homomorphisms. Simulation counterexamples are more
  complicated than case~3 when considering simulations relations that
  are not total functions.}  Each of the formulae can be directly
translated to a Boolean constraint over the vocabulary introduced in
Sect.~\ref{sec:guessing}. We illustrate how the first kind of
counterexample is handled, assuming $v = a_1\cdots a_m \in
\Sigma^*_\padding$ is the word in question; the two other cases are
similar. We introduce Boolean variables~$e_{i,q}$ for each $i \in
\{0,\ldots,m\}$ and state~$q \in \controls$, which will be used to
identify an accepting path in the transducer with input letters
corresponding to the word~$v$. We add constraints that ensure that
exactly one $e_{i,q}$ is set for each state~$q \in \controls$, and
that the path starts at the initial state~$q_0 = 1$ and ends in an
accepting state:
\begin{equation*}
  \Big\{ \bigvee_{q \in \controls} e_{i,q} \Big\}_{i \in \{0,\ldots,m\}},
  \quad
  \Big\{ \neg e_{i,q} \vee \neg e_{i,q'} \Big\}_{
    \substack{i \in \{0,\ldots,m\}\\
    q \not= q' \in \controls}},\quad
  e_{0, 1},\quad
  \big\{ e_{m, q} \to z_q \big \}_{q \in \controls}~.
\end{equation*}
For each $i \in \{1,\ldots,m\}$ a transition on the path, with input
letter~$a_i$ has to be enabled:
\begin{equation*}
    \Big\{
    e_{i-1,q} \wedge e_{i, q'} \to
    \bigvee_{b \in \ialphabet} x_{(q, a_i, b, q')}
    \Big\}_{\substack{i \in \{1,\ldots,m\}\\ q, q' \in \controls}}
    ~.
\end{equation*}

\section{Implementation and Evaluation}\label{sec:expr}
\anthonychanged{
We have implemented a prototype tool based on the aforementioned approach
for verifying and synthesising regular symmetry patterns.} 
The programming language is Java and we use
SAT4J~\cite{DBLP:journals/jsat/BerreP10} as the SAT solver. The source code and
the benchmarking examples
can be found at \url{https://bitbucket.org/truongkhanh/parasymmetry}. The input of our tool includes a model (i.e.~a textual representation of transducers), 
and optionally a set of finite instance symmetries (to speed up synthesis
of regular symmetry patterns), which can be generated using existing tools 
like~\cite{ASE13}.



\begin{table}[t]
\centering
\caption{Experimental Results on Verifying and Generating Symmetry Patterns }\label{exp1}
\begin{tabular}{@{\quad}l@{\quad}|@{\quad}*{3}{l@{\quad}}}
    \hline
    \textbf{Symmetry Systems (\#letters)} &
    \textbf{\# Transducer states} & 
    \textbf{Verif.\ time} & \textbf{Synth.\ time}\\\hline
	Herman Protocol (2) & 5 & 0.0s & 4s\\
	Israeli-Jalfon Protocol (2) & 5 & 0.0s & 5s\\
	Gries's Coffee Can (4) & 8 & 0.1s & 3m19s\\
	Resource Allocator (3) & 11 & 0.0s & 4m56s\\
	Dining Philosopher (4) & 17 & 0.4s & 26m\\\hline
\end{tabular}
\end{table}

We apply our tool to 5 models: the Herman self-stabilising protocol~\cite{Her90}, Israeli-Jalfon self-stabilising protocol~\cite{IJ90}, the Gries' coffee can example~\cite{Gries}, Resource Allocator, and Dining Philosopher. 
For the coffee can example,
the tool generates the functional symmetry pattern described in 
Section~\ref{sec:framework},
whereas the tool generates rotational process symmetries 
for the other models 
(see the \shortlong{full version}{appendix} for state diagrams).
Finite instance symmetries were added as constraints in the last
three examples.

Table~\ref{exp1} presents the experimental results: the number of
states of the synthesised symmetry transducer, the time needed to verify
that the transducer indeed represents a symmetry pattern (using the
method from Section~\ref{sec:verification}), and the total time needed
to compute the transducer (using the procedure from
Section~\ref{sec:onesym}). The data are obtained using a MacBook Pro
(Retina, 13-inch, Mid 2014) with 3 GHz Intel Core i7 processor and 16
GB 1600 MHz DDR3 memory.
In almost all cases, it takes less than 5 minutes (primarily SAT-solving) to 
find the regular symmetry patterns for all these models.
As expected, the verification step is quite
fast ($< 1$ second).


\section{Future work}
\label{sec:conc}
\OMIT{
We have described an expressive symbolic language based on finite-state 
letter-to-letter word transducers for capturing symmetry patterns for 
parameterised systems that enjoys nice algorithmic properties including
automatic symmetry verification and symmetry synthesis. The framework is
flexible in that constraints can be easily added/removed in the
verification/synthesis step as required by the user. 
}
Describe the expressivity and nice algorithmic properties that regular symmetry
patterns enjoy, we have pinpointed a limitation of regular symmetry patterns in 
expressing certain process symmetry patterns (i.e. reflections) and showed how 
to circumvent it by extending the framework to include symmetry patterns that 
can be recognised
by height-unambiguous pushdown automata. One possible future research direction
is to generalise our symmetry synthesis algorithm to this more general class of
symmetry patterns. Among others, this would require coming up with a syntactic 
restriction of this ``semantic'' class of pushdown automata. 

\paragraph{Acknowledgment:}
Lin is supported by Yale-NUS Grants, R\"ummer by the Swedish Research Council.
We thank Marty Weissman for a fruitful
discussion.

\bibliographystyle{abbrv}
\bibliography{references}

\shortlong{}{
\appendix

\section{Transducer Representation for Gries's Coffee Can Example}
    We shall first build basic transitions. Firstly, for each
    $i, j \in \N$, we define $\DEC_{i,\geq j}$ (resp. $\INC_{i, \geq j}$) to be
    a transducer over $\{1,\bot\}$ which checks that the current value (i.e.
    the number of 1s) is at least $j$ and decrements (resp. increments) it by
    $i$. For example, the transducer for $\DEC_{1,\geq 2}$ can be defined as 
    follows:
    \begin{center}
    \begin{tikzpicture}[%
    >=stealth,
    shorten >=1pt,
    node distance=2cm,
    on grid,
    auto,
    state/.append style={minimum size=2em},
    thick
  ]
    \node[state] (A)              {};
    \node[state] (B) [right of=A] {};
    \node[state,accepting] (C) [right of=B] {};

    \path[->] (A) +(-1,0) edge (A)
              (A)         edge              node {$1/1$} (B)
              (B)         edge [loop above] node {$1/1$} ()
              (B)         edge [bend left]  node {$1/\bot$} (C)
              (C)         edge [loop above] node {$\bot/\bot$} ();
  \end{tikzpicture}
  \end{center}
  For $\INC_{i,\geq j}$, we will make the assumption that if the buffer
  has reached its maximum capacity (i.e. there is not enough padding symbol 
  $\bot$ in the suffix), then the value stays as is. For example, the transducer
  for $\INC_{1,\geq 0}$ is:
    \begin{center}
    \begin{tikzpicture}[%
    >=stealth,
    shorten >=1pt,
    node distance=2cm,
    on grid,
    auto,
    state/.append style={minimum size=2em},
    thick
  ]
    \node[state,accepting] (B) {};
    \node[state,accepting] (C) [right of=B] {};

    \path[->] 
              (B) +(-1,0) edge (B)
              (B)         edge [loop above] node {$1/1$} ()
              (B)         edge [bend left]  node {$\bot/1$} (C)
              (C)         edge [loop above] node {$\bot/\bot$} ();
  \end{tikzpicture}
  \end{center}
  For each $z \in \{\x,\y\}$, we let $\DEC_{i,\geq j}[z]$ (resp.
  $\INC_{i,\geq j}[z]$) denote the transducer obtained from $\DEC_{i,\geq j}$
  (resp. $\INC_{i,\geq j}$) by replacing every occurrence of $1$ and $\bot$ by
  $1_z$ and $\bot_z$.

  The transducer for the example can then be obtained by connecting
  these basic components. For example, the transducer for case (b) is simply
  the transducer $\INC_{1,\geq 0} \circ \DEC_{2,\geq 2}$, where $\circ$ is
  simply the language concatenation operator (i.e. connecting the final state of
  $\INC_{1,\geq 0}$ with the initial state of $\DEC_{2,\geq 2}$). 

\section{Proof of Proposition \ref{prop:proj}}
    The proof of this proposition is a simple automata construction that
    relies on the fact that regular relations are closed under projections. 
    For simplicity, we shall illustrate this only in the 
    case when 
    $S = \ialphabet^*$ (for a general regular set $S$, we simply need to
    do a simple product construction). Suppose that the transducer for 
    $\to$ is $\Aut$ and 
    that the transducer for $R$ is
    $\Aut' = (\ialphabet_\padding^2,\controls',\transrel',q_0',\finals')$. 
    Let $\struct_1 = \langle R(S); \to_1 \rangle$ be the image of $\struct$
    under $R$. Let $\lambda: (\ialphabet \cup \{\epsilon\}) \to
    \ialphabet_\padding$ be a function such that $\lambda(a) = a$, for each 
    $a \in \ialphabet$, and $\lambda(\epsilon) = \padding$.
    The set $R(S)$ is regular and can be constructed by a simple
    projection operation of $\transrel'$ onto the second component, i.e.,
    define the new transition relation $\transrel_1' \subseteq \controls'
    \times (\ialphabet \cup \{\epsilon\}) \times \controls'$ such that,
    for each $a \in \ialphabet \cup \{\epsilon\}$, $(p,a,q) \in \transrel_1'$ 
    iff $(p,(b,\lambda(a)),q) \in \transrel'$ for some $b \in
    \ialphabet$. The 
    transducer for $\to_1$ can then be obtained by restricting $\to$
    to $R(S)$ by means of product automata construction.  The construction of 
    automata presentation of $\struct_1$ in fact runs in polynomial-time. 

\section{Finishing proof of Theorem \ref{th:paramreg}}
    We first show that $\mathcal{F}$ is effectively regular. Given an alphabet
    $\ialphabet$, the function $f_{\mathcal{F}}$ maps the word 
    $a_1\ldots a_n \in \ialphabet^n$ to $a_2a_1a_3a_4\ldots a_n$, i.e., the 
    first two letters
    are swapped. The transducer for $f_{\mathcal{F}}$ is
    $(\ialphabet^2,\controls,\transrel,q_0,\finals)$, where $\controls =
    \{q_0,p\} \cup \{q_v : v \in \ialphabet^2\}$, $\finals = \{ p \}$, and 
    $\transrel$ consists of the following transitions: (1) $(q_0,v,q_v)$ for 
    each $v \in \ialphabet^2$, (2) $(q_{(a,b)},(b,a),p)$ for each 
    $(a,b) \in \ialphabet^2$, and (3) $(p,(a,a),p)$ for each $a \in 
    \ialphabet$. This construction is effective (in fact, can be performed
    in $O(|\ialphabet|^2)$.

    We now show that $\mathcal{F}'$ is effectively regular. Given an alphabet
    $\ialphabet$, the function $f_{\mathcal{F}'}$ maps the word 
    $a_1\ldots a_n$ to $a_n a_1 a_2 \ldots a_{n-1}$. The transducer for 
    $f_{\mathcal{F}'}$ is $(\ialphabet^2,\controls,\transrel,q_0,\finals)$,
    where $\controls = \{q_0\} \cup \{ q_v : v \in \ialphabet^2 \}$,
    $\finals = \{q_F\}$, and $\transrel$ consists of the following
    transitions: (1) $(q_0,v,q_v) \in \transrel$ for each $v\in \ialphabet^2$,
    and (2) $(q_{(a,b)},(a',a),q_{(a',b)}) \in \transrel$ for each $a' \in 
    \Sigma$, and (3) $(q_{(a,b)},(b,a),q_F)$ for each $(a,b) \in \ialphabet^2$.

\section{Proof of Proposition \ref{prop:dihedralreg}}
    The proof is by contradiction. Suppose that $\dihgroup = \{D_n\}_{n \geq 2}$
    are generated by finitely many regular parameterised permutations 
    $\bar\pi^1, \ldots,\bar\pi^k$, say, with
    $\bar\pi^j = \{\pi^j_n\}_{n \geq 2}$. Let $\ialphabet = \{a,b\}$. Then,
    for each $j\in[k]$, there is a transducer $\Aut_j$ recognising 
    $f_{\pi^j}: \ialphabet^* \to \ialphabet^*$ such that 
    $f_{\pi^j}(\vecV) = \pi^j_n\vecV$ whenever $\vecV \in \ialphabet^n$.
    Let $m$ the maximum number of states among all the transducers $\Aut_j$.
    Let $N := (m+1)^3$, which as we shall see will make our pumping argument 
    works.

    Let $r = (1,2,\ldots,n)$ and $s = (1,n)(2,n-1)\cdots (\lfloor
    n/2\rfloor,\lceil n/2\rceil)$. Each instance $D_n$ in $\dihgroup = 
    \{D_n\}_{n \geq 2}$ is generated by $r$ and $s$. The group 
    consists of the elements
    \[
        r^0, r^1, \ldots, r^{n-1}, s, sr^1, \ldots, sr^{n-1}.
    \]
    Any generating set $X$ for $D_n$ must contains at least one of
    $s, sr^1, \ldots, sr^{n-1}$. Let $M = 4N$. Then, for some $j \in [k]$,
    the permutation $\pi^j_M$ is $sr^l$ for some $l \in [0,M-1]$. There
    are two cases:
    \begin{itemize}
        \item $l$ is even, i.e., $l = 2h$ for some integer $h$. Then, 
            \[ 
                sr^l = (h,h+1)(h-1 \pmod{M},h+2)\cdots (h-\frac{M}{2}
                \pmod{M},h+\frac{M}{2}+1).
            \]
            Visually, this symmetry
            is simply a reflection symmetry on an $M$-gon with vertices 
            labeled $1,\ldots,M$ along the line connecting the 
            middle point between $h$ and $h+1$ and the center of the polygon.
        \item $l$ is odd, i.e., $l = 2h+1$ for some integer $h$. Then, 
            \[
                sr^l = (h,h+2)(h-1 \pmod{M},h+3)\cdots (h-\frac{M}{2}-1
                \pmod{M},h+\frac{M}{2}+1).
            \]
            Visually,
            this symmetry is simply a reflection symmetry on an $M$-gon with 
            vertices labeled $1,\ldots,M$ along the line connecting 
            $h+1$ and the center of the polygon.
    \end{itemize}
    In the following, we will consider only the first case since the second 
    case can be dealt with using the same argument. In this case, observe that
    $sr^l$ maps $h+N+1$ to $h-N \pmod{M}$ and vice versa. We assume that 
    $h+N+1 < h-N \pmod{M}$. The other case can be dealt with using the same
    argument. In this case, let $t_1 = h+N+1$ and $t_2 = h-N \pmod{M}$.
    Consider the word $w = a_1\cdots a_M \in \ialphabet^*$ such that 
    \[
        a_i = \left\{ \begin{array}{cc}
                (b,a) & \text{ if $i = t_1$, } \\
                (a,b) & \text{ if $i = t_2$, } \\
                (a,a) & \text{ otherwise. }
                      \end{array}
                \right.
    \]
    Therefore, $w = w_1(b,a)w_2(a,b)w_3$, for some $w_1, w_2, w_3 \in (a,a)^*$.
    Notice also that by definition $|w_2| = N$ and $|w_1| + |w_3| = 3N-2$.
    There are two cases: $|w_1| < |w_3|$ or $|w_1| \geq |w_3|$. We shall
    consider only $|w_1| < |w_3|$; the other case can be treated in the same
    way. Consider an accepting run $\sigma$ of $\Aut_j$ on $w$. The run can
    be divided into five segments: $\sigma = \sigma_1 \odot \gamma_1 \odot 
    \sigma_2 
    \odot \gamma_2 \odot \sigma_3$, where $\sigma_i$ is a run segment of 
    $\Aut_j$ on $w_i$ (when $i = 1$, from the initial state),
    $\gamma_1$ is a transition on input $(b,a)$, and 
    $\gamma_2$ is a transition on input $(a,b)$.

    In the following, we shall show that, for some $w_2', w_3' \in (a,a)^*$ with
    $w_i' \neq w_i$, the word $w_1(b,a)w_2'(a,b)w_3'$ is also accepted
    by $\Aut_j$, contradicting that $f_{\bar \pi^j}$ is a bijective function.
    This is done by a pumping argument. Since $|w_2| > m$, then a state
    $q$ of the transducer is visited twice in $\sigma_2$ (e.g. after consuming
    $m' \leq m$ number of $a$s). Let us now analyse the path $\sigma_3$,
    which is of length $> N$. Let us further subdivide $\sigma_3$ into
    at least $m^2$ contiguous segments of length $m+1$. Since one state
    of the transducer is visited twice in each segment (after consuming
    at most $m$ number of $a$s) and that there are at least $m^2$ disjoint 
    segments, we can infer that at least $m'$ of these segments have cycles
    of the same length $L$. We can construct $w_2'$, $w_3'$ and their
    accepting runs $\sigma_2'$ and $\sigma_3'$ by simply removing $m'$ cycles 
    of length $L$ from $\sigma_3$ and repeating the cycle of length 
    $m'$ in $\sigma_2$ precisely $L$ times.

\section{Boolean encoding of the conditions in Section~\ref{sec:guessing}}
\label{app:encoding}

\paragraph{\textbf{(C3)}, \textbf{(C4)}} We assume additional
variables~$y_q, y_q'$ (for each $q \in \controls$) ranging over the
interval $[0,n-1]$ to encode the distance of state from an accepting
or from the initial state. The following formulae (instantiated for
each $q \in \controls$) define the value of the variables, and imply
that every state is only finitely many transitions away from the
initial and some accepting state:
\begin{align*}
  & z_q \to y_q = 0, &&
  \neg z_q \to \bigvee_{a, b \in \ialphabet, q' \in \controls}
  \big(x_{(q, a, b, q')} \wedge y_q = y_{q'} + 1\big)
  \\
  & q = 1 \to y'_q = 0,
  && q \not= 1 \to \bigvee_{a, b \in \ialphabet, q' \in \controls}
  \big(x_{(q', a, b, q)} \wedge y'_q = y'_{q'} + 1\big)
\end{align*}
In the third constraint, recall that $q_0 = 1$ is the initial state.
The integer variables can further be encoded as a vector of Boolean
variables and unary or binary representation.

\paragraph{\textbf{(C6)}}

Since we search for transducers in which every state is reachable from
the initial state, and in which an accepting state is reachable from
every state, \textbf{(C6)} reduces to the existence of some cycle in
the automaton. Because the transducer is furthermore supposed to
represent a total function, the input letters of the transitions on
the cycle can even be fixed to an arbitrary letter~$a_0 \in
\ialphabet$. To encode the condition, we assume additional Boolean
variables $c_{(q, b, q')}$ for each tuple $(q,b,q') \in \controls
\times \ialphabet \times \controls$; the assignment~$c_{(q, b, q')} =
1$ will represent the fact that the transition~$t = (q, a_0, b, q')$
occurs on some infinite path in the automaton (i.e., on a path on
which at least one state occurs infinitely often). This is expressed
by the following set of Boolean formulae, stating that at least one
$c_{(q, b, q')}$ is set, that each $c_{(q, b, q')}$ implies that the
corresponding transition is present in the transducer, and that each
$c_{(q, b, q')}$ has a successor transition:
\begin{gather*}
  \bigvee_{(q, b, q')} c_{(q, b, q')},
  \qquad
  \Big\{ c_{(q, b, q')} \to x_{(q, a_0, b, q')}
  \Big\}_{q, q' \in \controls, b \in \ialphabet}~,
  \\
  \Big\{ c_{(q, b, q')} \to \bigvee_{q'' \in \controls, b' \in \ialphabet} c_{(q', b', q'')}
  \Big\}_{q, q' \in \controls, b \in \ialphabet}
\end{gather*}

\paragraph{\textbf{(C11)}}
The main idea when encoding \textbf{(C11)} as a set of Boolean
formulae is to consider the \emph{Parikh difference} between input and
output letters on paths in the transducer~$\Aut$; importantly, we can
observe that this difference, for paths starting in the initial state,
is independent of the path itself and only depends on the last state
of the path:
\begin{lemma}
  Suppose~$\Aut$ is a deterministic transducer satisfying
  \textbf{(C4)} and \textbf{(C11)}, and let $\rho_1$ and $\rho_2$ be
  paths in $\Aut$ that start in the initial state~$q_0$ and end in the
  same (not necessarily accepting) state~$q$. Let $v_1 \otimes w_1$
  and $v_2 \otimes w_2$ be the words generated by $\rho_1$ and
  $\rho_2$, respectively. Then $\Parikh{w_1} - \Parikh{v_1}
  = \Parikh{w_2} - \Parikh{v_2}$.
\end{lemma}
\begin{proof}
  Thanks to \textbf{(C4)}, there is a path~$\rho$ from $q$ to some
  accepting state; i.e., $\rho_1 \cdot \rho$ and $\rho_2 \cdot \rho$
  are two accepting paths of $\Aut$. Let $v_1\cdot v \otimes w_1 \cdot
  w$ and $v_2\cdot v \otimes w_2 \cdot w$ be the corresponding
  accepted words. Due to \textbf{(C11)}, we have $\Parikh{v_i\cdot v}
  = \Parikh{w_i \cdot w}$ for $i \in \{1, 2\}$, and thus
  $\Parikh{v_i\cdot v} - \Parikh{w_i \cdot w} = \Parikh{v_i}
  + \Parikh{v} - \Parikh{w_i} - \Parikh{w} = 0$, which implies the
  conjecture.  \qed
\end{proof}

This means that every state~$q$ in a transducer~$\Aut$ satisfying
\textbf{(C11)} has a fixed vector~$d_q$ of Parikh differences, which
can be determined by looking at an arbitrary path from the initial
state to $q$ (and then computed as $d_q = \Parikh{w} - \Parikh{v}$,
for the word~$v \otimes w$ accepted by the path). The initial and
every accepting state has Parikh difference~$d_q = 0$.

This leads to a simple encoding of \textbf{(C11)}: we introduce
variables~$d_q^a$ for every state~$q \in \controls$ and every
character~$a \in \ialphabet$, ranging over the set~$\{-n+1, \ldots,
n-1\}$ of integers. Variable~$d_q^a$ will represent the Parikh
difference for letter~$a$ in state~$q \in \controls$, which is
uniquely determined by a set of constraints about the Parikh difference for
initial and final states, as well as the effect of transitions:
\begin{gather*}
  \big\{ d_{q_0}^a = 0 \big\}_{a \in \ialphabet}~,\quad
  \big\{ z_q \to d_q^a = 0\big\}_{a \in \ialphabet,\; q \in \controls}~,
  \\
  \Big\{ x_{(q, a, b, q')} \to \big(
  d_{q'}^a = d_q^a + 1 \wedge d_{q'}^b = d_q^b - 1
  \big) \Big\}_{\substack{q, q' \in \controls \\ a \not= b \in \ialphabet}}~,
  \\
  \Big\{ x_{(q, a, a, q')} \to   d_{q'}^a = d_q^a
  \Big\}_{\substack{q, q' \in \controls\\ a \in \ialphabet}}~,\quad
  \Big\{ x_{(q, a, b, q')} \to d_{q'}^c = d_q^c
  \Big\}_{\substack{q, q' \in \controls\\ a, b \in \ialphabet \\ c \in \ialphabet\setminus\{a, b\}}}~.
\end{gather*}

\section{Descriptions of protocols}
\label{sec:protocols}

\subsection{Herman's protocol}
    Here is an example based on a nondeterministic version of Herman's
    self-stabilising protocol (only for odd numbers of processes). Again, it is
    very similar to before (i.e. in a ring) and passes tokens around. The variation
    here is that a process can pass a token
    to its right neighbour or keep it, but not to the left neighbour. Two tokens
    are merged as soon as they are held by the same process (i.e. same as before).
    A formalisation of
    this protocol as a transducer goes as before, i.e., it is a union of:
\begin{itemize}
\item $I^* \binom{\bot}{\top}\binom{\top}{\bot} I^*$
\item $I^* \binom{\top}{\top}\binom{\top}{\bot} I^*$
\item $\binom{\top}{\bot}I^*\binom{\bot}{\top}$
\item $\binom{\top}{\bot}I^*\binom{\top}{\top}$
\end{itemize}
    Notice that here you only have the rotation symmetry, but not the reflection
    symmetry.

    \OMIT{
\subsection{A variant of the producer-consumer protocol with buffers}
    This example is motivated by producer-consumer problem that \emph{forgets}
    to take into account the integer buffer. We have an
    integer buffer with $k$ bits, $m$ producers, and $n$ consumers. A producer
    has two states: an idle state $I_P$ and a produce state $P$. Similarly,
    a consumer has two states: an idle state $I_C$ and a consume state $C$.
    At any given step, a producer or a consumer is chosen. Then, if it is in
    idle state, it goes to state $P$ and increment the buffer by
    1 (if it overflows, it simply goes back to 0). If it is in the state $P$,
    it will go back to state $I_P$ (if chosen). The case of consumer is defined
    in a similar way, but if it's chosen it can go to state $C$ and decrement the
    buffer by 1 (again, the buffer goes to $1^k$ if initially it's $0^k$).

    We can implement this example as a transducer as follows. We represent
    the buffer starting with the least significant digit, i.e., in a reverse
    order. For example, the number 8 will be represented as 0001. 
}	

\section{Transducers Found From Our Program}
Figures~\ref{example1}-\ref{coffeecan} show the transducers generated from our 
program for all the systems used in our experiments.

\begin{figure}[t]
\centering
\includegraphics[scale = 0.5]{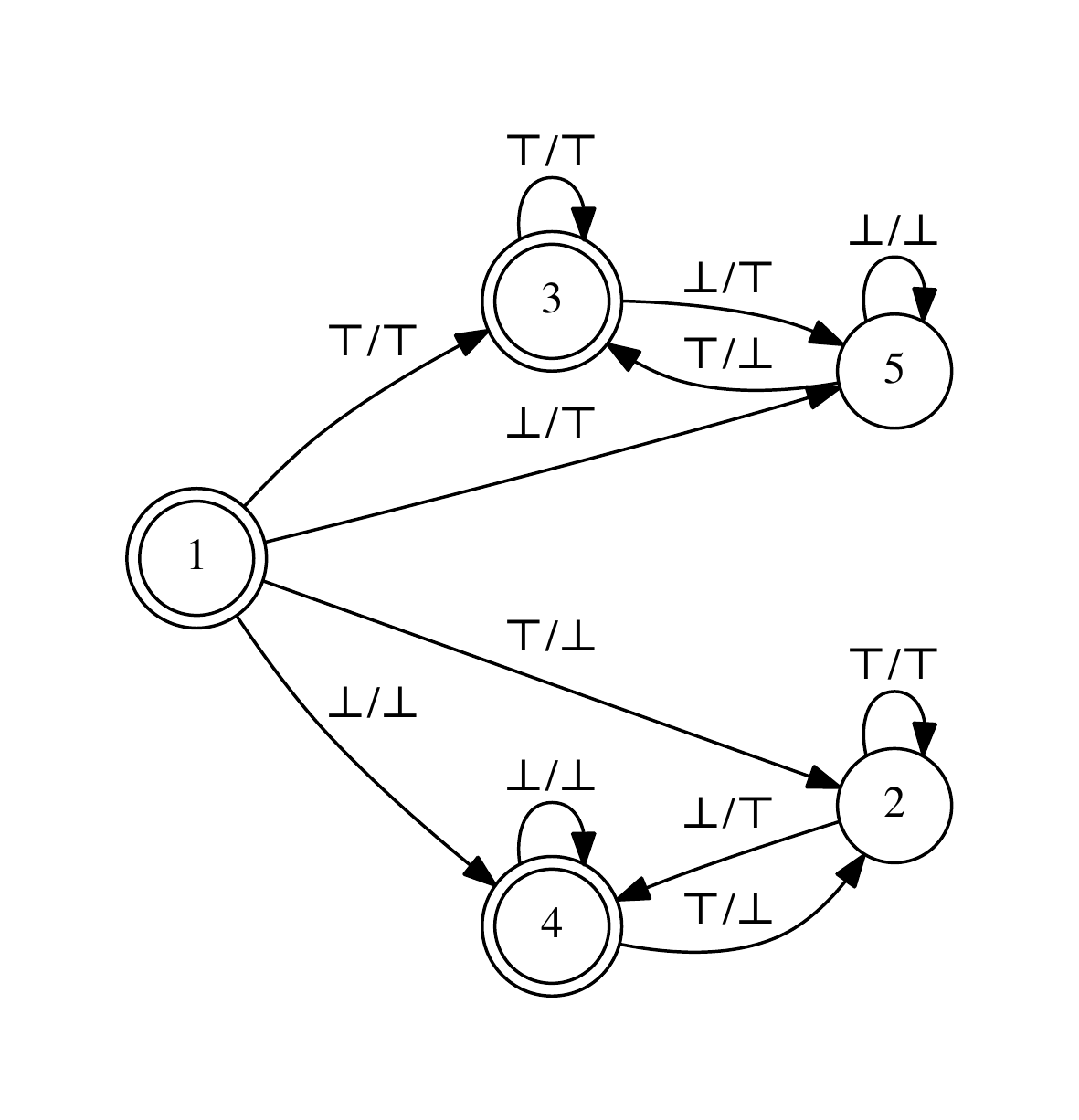}
\caption{Transducer of Herman Protocol Example}
\label{herman}
\end{figure}

\begin{figure}[t]
\centering
\includegraphics[scale = 0.5]{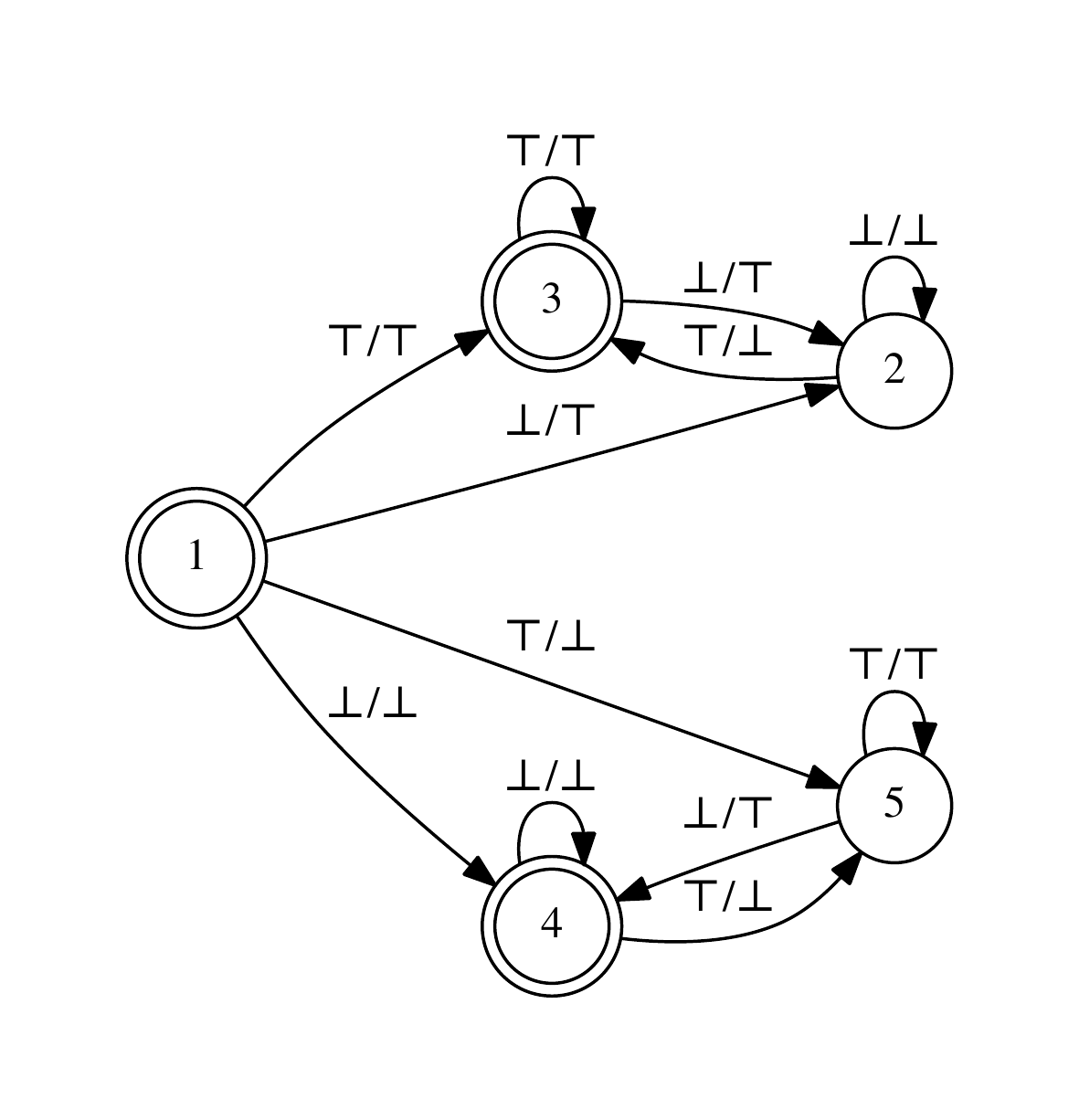}
\caption{Transducer of Israeli-Jalfon Protocol Example}
\label{example1}
\end{figure}

\begin{figure}[t]
\centering
\includegraphics[scale = 0.5]{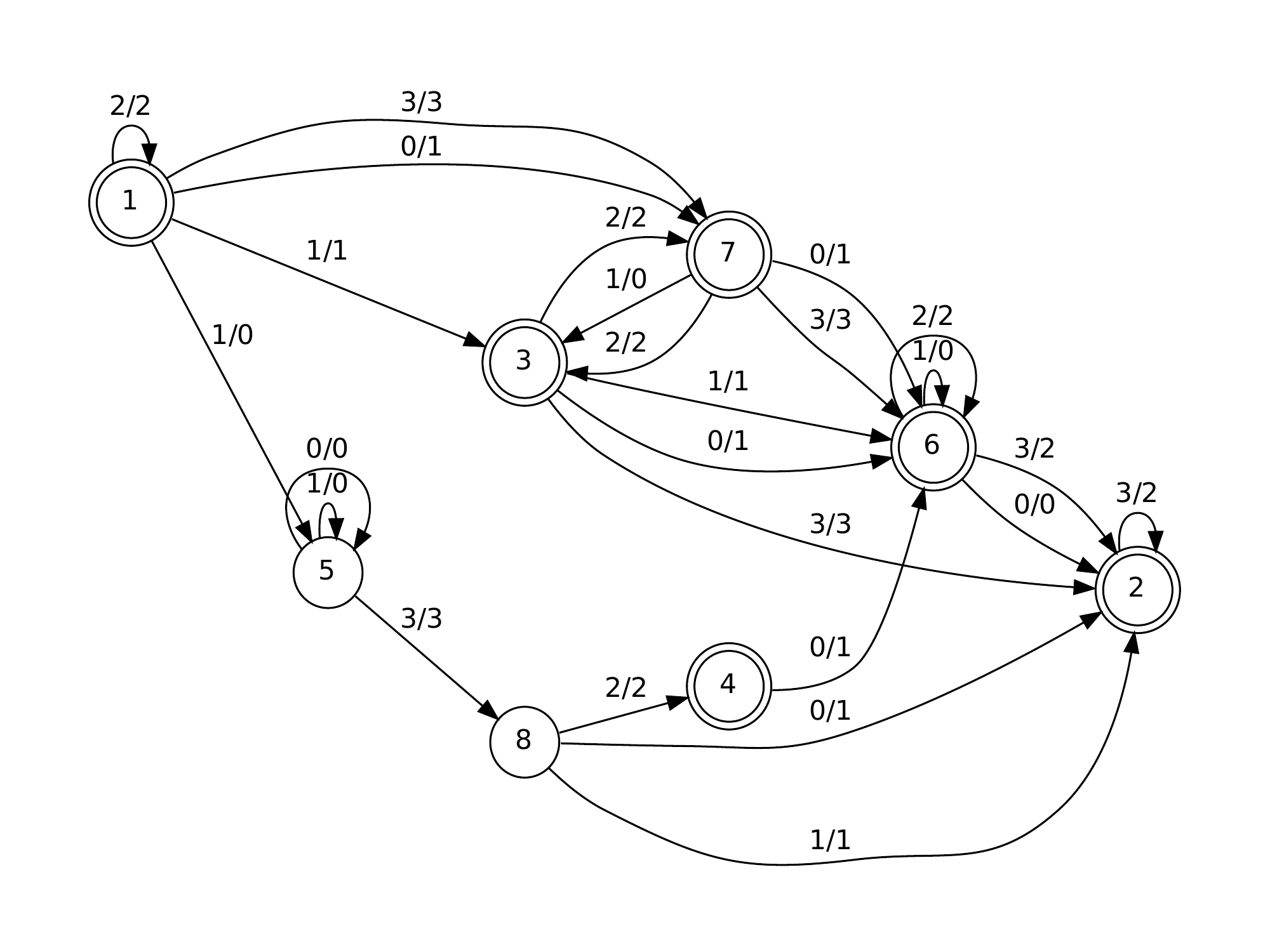}
\caption{Transducer of Gries’s Coffee Can Example}
\label{coffeecan}
\end{figure}

\begin{figure}[t]
\centering
\includegraphics[scale = 0.9]{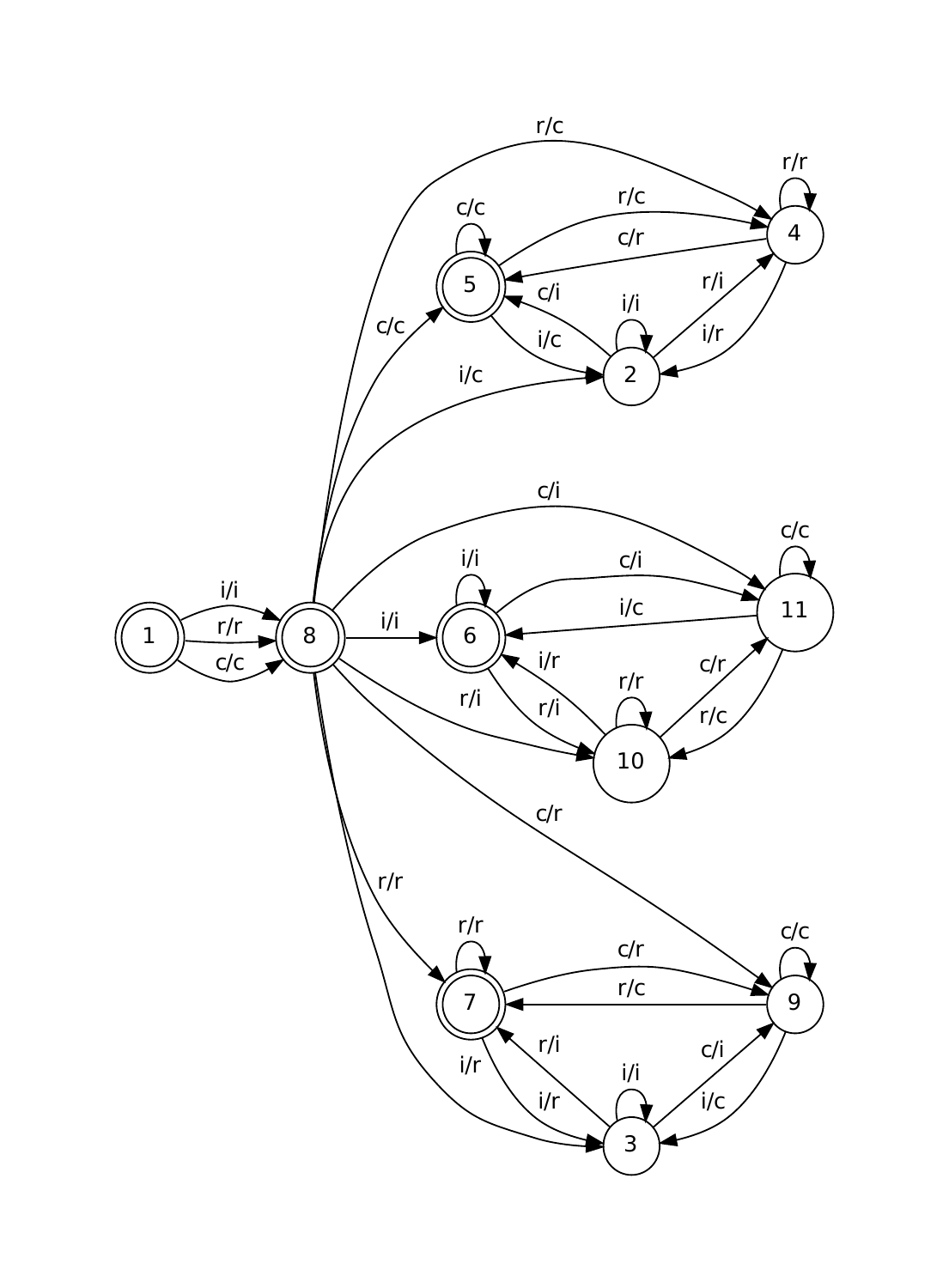}
\caption{Transducer of the Resource Allocator Example}
\label{consumerproducer}
\end{figure}

\begin{figure}[t]
\centering
\includegraphics[scale = 1.2]{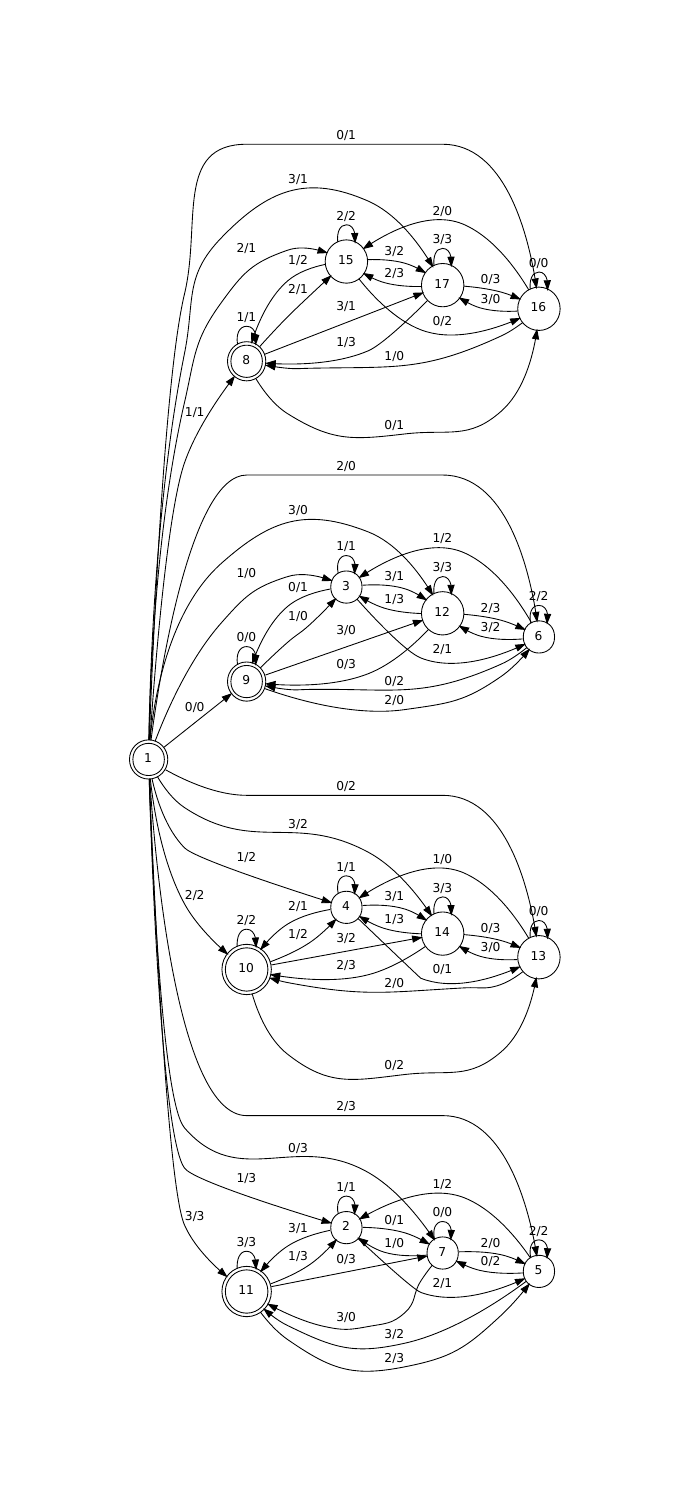}
\caption{Transducer of Dining Philosopher Example}
\label{DP}
\end{figure}

}

\end{document}